\documentclass[reqno]{amsart}

\usepackage{mathtools,amsmath,amsthm,amssymb,amsfonts,enumerate,graphicx,color,pstricks,rotating}
\usepackage[T1]{fontenc} 
\numberwithin{equation}{section} 
\theoremstyle{plain}
\newtheorem{theo+}           {Theorem}      [section]
\newtheorem{prop+}  [theo+]  {Proposition}
\newtheorem{coro+}  [theo+]  {Corollary}
\newtheorem{lemm+}  [theo+]  {Lemma}
\newtheorem{defi+}  [theo+]  {Definition}
\newtheorem{conj+}  [theo+]  {Conjecture}

\theoremstyle{definition}
\newtheorem{rema+}  [theo+]  {Remark}
\newtheorem{prob+}  [theo+]  {Problem}
\newtheorem{exam+}  [theo+]  {Example}

\newenvironment{theorem}{\begin{theo+}}{\end{theo+}}
\newenvironment{proposition}{\begin{prop+}}{\end{prop+}}
\newenvironment{corollary}{\begin{coro+}}{\end{coro+}}
\newenvironment{lemma}{\begin{lemm+}}{\end{lemm+}}

\newcommand{\om}{\omega}
\newcommand{\tha}{\theta}
\newcommand{\ti}{\textup i}

\newcommand{\id}{\operatorname{id}}
\begin{document}

\baselineskip 18pt
\larger[2]
\title[Special polynomials related to the eight-vertex model III]
{Special polynomials related to the supersymmetric eight-vertex model. III. Painlev\'e VI equation.} 
\author{Hjalmar Rosengren}
\address
{Department of Mathematical Sciences
\\ Chalmers University of Technology and University of Gothenburg\\SE-412~96 G\"oteborg, Sweden}
\email{hjalmar@chalmers.se}
\urladdr{http://www.math.chalmers.se/{\textasciitilde}hjalmar}
 %\keywords{Three-colour model, eight-vertex-solid-on-solid model, domain wall boundary conditions, partition function, alternating sign matrix}
%\subjclass[2000]{05A15, 33E05, 82B23}

\thanks{Research  supported by the Swedish Science Research
Council (Vetenskapsr\aa det)}

\begin{abstract}
We prove that certain  polynomials previously introduced by the author can be identified with tau functions of Painlev\'e VI, obtained from one of Picard's algebraic solutions by acting with a four-dimensional lattice
of B\"acklund transformations. For particular lines in the lattice, this proves conjectures of Bazhanov and Mangazeev. As applications, we describe the  behaviour of the corresponding solutions  near the singular points of Painlev\'e VI,
and obtain several new properties of our polynomials. 
\end{abstract}

\maketitle        

% \vspace*{5mm}
% \begin{center}
% {\Huge\bf Draft 21 October 2011}
% \end{center}
% \vspace*{5mm}

\section{Introduction}

The present work is the third part of a series, devoted to the study of  special polynomials related to the eight-vertex model and other solvable lattice models of statistical mechanics. In \cite{r2} we introduced, for each non-negative integer $m$, a 
four-dimensional lattice $T_n^{(k_0,k_1,k_2,k_3)}$ of symmetric rational functions in $m$ variables, depending also on a parameter $\zeta$. Here, $k_j$ and $n$ are integers, such that $m+\sum_j k_j=2n$. 
The denominator in these functions is elementary, so they are essentially
symmetric polynomials. For $m=0$ and $m=1$, polynomials corresponding to particular lines in the lattice appear in various ways in connection with 
 solvable models 
 \cite{bm1,bm2,bh,fh,h,mb,ras,r0,r1,zj}.
In  \cite{r3}, we proved that the polynomials satisfy a
non-stationary Schr\"odinger equation, which can be considered as the
canonical quantization of Painlev\'e VI. 

In the present work,  we will show that
the case $m=0$ of the polynomials can be identified
 with
tau functions of Painlev\'e VI, obtained from one of Picard's algebraic solutions by acting with a four-dimensional lattice of B\"acklund transformations. For particular lines in the lattice, this has been conjectured by
 Bazhanov and Mangazeev \cite{bm2}. 

The plan of the paper is as follows. In \S \ref{ps}, we recall the relevant facts on  Painlev\'e VI. In particular, we must understand the action of B\"acklund transformations on tau functions. Although that topic has been considered by Masuda \cite{m}, his conventions are not ideal for our purposes and we therefore rederive some
of his results in  slightly different form. In \S \ref{sss}, we 
consider
the   tau functions corresponding to one of Picard's algebraic solutions,
realizing  them explicitly as modular functions. 
After these preliminaries, we can turn to our main result,
Theorem \ref{trt}, 
which relates Painlev\'e tau functions to the case $m=0$ of our polynomials. In \S \ref{aps} we give some applications. Using results of
 \cite{r2} we  describe the
behaviour of the corresponding four-dimensional lattice of algebraic solutions to Painlev\'e VI at the singular points of the equation, see Corollary \ref{scc}.  We  also  obtain a new symmetry for our polynomials, Corollary \ref{nsc}.
Reformulating bilinear identities for tau functions in terms of our polynomials, we can
 prove recursions along particular lines in the lattice conjectured by Bazhanov and Mangazeev \cite{bm2,mb}, see Proposition~\ref{nbp} and the subsequent discussion. Finally, we observe that the 
 $\operatorname{E_{VI}}$ equation for the Hamiltonian of Painlev\'e VI
leads to   quadratic differential equations for our polynomials, see
Proposition \ref{qdp2}.

\section{Painlev\'e  VI}
\label{ps}

\subsection{B\"acklund transformations}
 Painlev\'e VI is the  differential equation
\begin{align}\notag\frac{d^2q}{dt^2}&=\frac 12\left(\frac 1q+\frac 1{q-1}+\frac 1{q-t}\right)
\left(\frac{dq}{dt}\right)^2-\left(\frac 1t+\frac 1{t-1}+\frac 1{q-t}\right)\frac{dq}{dt}\\ 
\label{py}&\quad +\frac{q(q-1)(q-t)}{t^2(t-1)^2}\left(\alpha+\beta\frac t{q^2}+\gamma\frac{t-1}{(q-1)^2}+\delta\frac{t(t-1)}{(q-t)^2}\right).\end{align}
It is the most general Painlev\'e  equation, and appears in many areas of
 contemporary mathematics and physics. 

We will briefly review the rich symmetry theory of
 Painlev\'e VI. It is mainly due to Okamoto \cite{o}, although we will follow the  exposition of Noumi and Yamada \cite{ny}.
We introduce parameters $\alpha_0,\dots,\alpha_4$ satisfying the constraint
\begin{equation}\label{ac}\alpha_0+\alpha_1+2\alpha_2+\alpha_3+\alpha_4=1
\end{equation}
and related to the parameters of \eqref{py} by
$$\alpha=\frac{\alpha_1^2}{2},\qquad \beta=-\frac{\alpha_4^2}{2},\qquad
\gamma=\frac{\alpha_3^2}{2},\qquad \delta=\frac{1-\alpha_0^2}{2}.
 $$
We let 
\begin{multline*}H=q(q-1)(q-t)p^2-\big\{(\alpha_0-1)q(q-1)+\alpha_3q(q-t)+\alpha_4(q-1)(q-t)\big\}p\\
+\alpha_2(\alpha_1+\alpha_2)(q-t).\end{multline*}
Then, \eqref{py} is equivalent to the Hamiltonian system
\begin{equation}\label{ph}t(t-1)\frac{dq}{dt}=\frac{\partial H}{\partial p},\qquad t(t-1)\frac{dp}{dt}=-\frac{\partial H}{\partial q}. \end{equation}

The system \eqref{ph} admits many symmetries, or \emph{B\"acklund transformations}. Indeed, it is invariant under the involutions
 $s_j$,  $r_j$ and  $t_j$ 
defined in the following table.
\vspace{1ex}
\begin{center}
\begin{tabular}{|c||ccccc|c c c|}
\hline
 & $\alpha_0$ & $\alpha_1$ & $\alpha_2$ & $\alpha_3$ & $\alpha_4$ & $q$ & $p$ & $t$\\
\hline \hline
&&&&&&&\\[-4mm]
$s_0$ & $-\alpha_0$ & $\alpha_1$ & $\alpha_2+\alpha_0$ & $\alpha_3$ & $\alpha_4$ & $q$ & $p-\frac{\alpha_0}{q-t}$ & $t$\\
$s_1$ & $\alpha_0$ & $-\alpha_1$ & $\alpha_2+\alpha_1$ & $\alpha_3$ & $\alpha_4$ & $q$ & $p$ & $t$\\
$s_2$ & $\alpha_0+\alpha_2$ & $\alpha_1+\alpha_2$ & $-\alpha_2$ & $\alpha_3+\alpha_2$ & $\alpha_4+\alpha_2$ & $q+\frac{\alpha_2}p$ & $p$ & $t$\\
$s_3$ & $\alpha_0$ & $\alpha_1$ & $\alpha_2+\alpha_3$ & $-\alpha_3$ & $\alpha_4$ & $q$ & $p-\frac{\alpha_3}{q-1}$ & $t$\\[1mm]
$s_4$ & $\alpha_0$ & $\alpha_1$ & $\alpha_2+\alpha_4$ & $\alpha_3$ & $-\alpha_4$ & $q$ & $p-\frac{\alpha_4}{q}$ & $t$\\[1mm]
$r_1$ & $\alpha_1$ & $\alpha_0$ & $\alpha_2$ & $\alpha_4$ & $\alpha_3$ & $\frac{t(q-1)}{q-t}$ & $\!\!\frac{(t-q)((q-t)p+\alpha_2)}{t(t-1)}$ & $t$\\[1mm]
$r_3$ & $\alpha_3$ & $\alpha_4$ & $\alpha_2$ & $\alpha_0$ & $\alpha_1$ & $\frac{t}{q}$ & $-\frac{q(pq+\alpha_2)}{t}$ & $t$\\[1mm]
% $r_4$ & $\alpha_4$ & $\alpha_3$ & $\alpha_2$ & $\alpha_1$ & $\alpha_0$ & $\frac{q-t}{q-1}$ & $\!\!\frac{(q-1)((q-1)p+\alpha_2)}{t-1}$ & $t$\\
$t_1$ & $\alpha_0$ & $\alpha_4$ & $\alpha_2$ & $\alpha_3$ & $\alpha_1$ & $\frac 1q$ & $-q(pq+\alpha_2)$ & $\frac 1t$\\
$t_3$ & $\alpha_0$ & $\alpha_1$ & $\alpha_2$ & $\alpha_4$ & $\alpha_3$ & $1-q$ & $-p$ & $\!\!\!\!1-t$\\
\hline
\end{tabular}
\end{center}
\vspace{1ex}

We will write $r_4=r_1r_3=r_3r_1$.
We consider these symmetries as 
 automorphisms of the differential
 field $\mathcal F_0$ generated by $\alpha_j$, $q$, $p$ and $t$, subject to the relation \eqref{ac}, and equipped with the derivation
$$\delta=\frac{\partial H}{\partial p}\frac{\partial}{\partial q}-\frac{\partial H}{\partial q}\frac{\partial}{\partial p}+t(t-1)\frac{\partial}{\partial t}. $$
In general,
a B\"acklund transformation can be defined as a field automorphism $\sigma$ such that
$$\sigma(\delta(x))\delta(\sigma(y))=\sigma(\delta(y))\delta(\sigma(x)),\qquad x,y\in\mathcal F_0. $$
Choosing, without loss of generality, $y=t$, we find that the 
 B\"acklund property means that, for all $k$, 
\begin{subequations}\label{bd}
\begin{align}
\label{bda}s_k\circ\delta&=\delta\circ s_k,& r_k\circ\delta&=\delta\circ r_k,\\
\label{bdb}t_1\circ\delta&=\frac 1t\,\delta\circ t_1,& t_3\circ\delta&=-\delta\circ t_3. \end{align}
\end{subequations}

The B\"acklund transformations defined above satisfy the relations
\begin{subequations}\label{wr}
\begin{align}
\label{wa}s_i^2&=1,& i=0,1,2,3,4,\\
\label{wb}(s_is_j)^2&=(s_is_2)^3=1, & i,j=0,1,3,4,\\
\label{wc}r_1^2&=r_3^2=(r_1r_3)^2=1, \\
\label{we}r_1s_{0,1,2,3,4}&=s_{1,0,2,4,3}r_1,&
r_3s_{0,1,2,3,4}&=s_{3,4,2,0,1}r_3,&\\
\label{wd}t_1^2&=t_3^2=(t_1t_3)^3=1,&\\ 
\label{wh}t_1s_{0,1,2,3,4}&=s_{0,4,2,3,1}t_1,&
t_3s_{0,1,2,3,4}&=s_{0,1,2,4,3}t_3,&\\
\label{wi}t_1r_{1,3,4}&=r_{4,3,1}t_1, &
t_3r_{1,3,4}&=r_{1,4,3}t_3.
\end{align}
\end{subequations}
In particular,  $(s_j)_{j=0}^4$ generate 
 the $D_4$ affine Weyl group. Adjoining
$(r_j)_{j=1,3}$ gives the \emph{extended} $D_4$ affine Weyl group.
 The elements $t_j$  generate the symmetric group $\mathrm S_3$; adjoining
them gives the extended $F_4$  affine Weyl group. 

The group of B\"acklund transformations contains a  subgroup isomorphic to $\mathbb Z^4$, corresponding to the $D_4$ weight lattice.
It is generated by the mutually commuting elements 
\begin{align*}T_1&=r_1s_1s_2s_3s_4s_2s_1,\qquad T_2=s_0s_2s_1s_3s_4s_2s_1s_3s_4s_2,\\
T_3&=r_3s_3s_2s_1s_4s_2s_3,\qquad T_4=r_4s_4s_2s_1s_3s_2s_4.\end{align*}
We will need the commutation relations
\begin{subequations}\label{sbtr}
\begin{align}\label{s0tr}s_0T_j&=\begin{cases}T_jT_2^{-1}s_0, & j=1,3,4,\\
T_2^{-1}s_0, & j=2,\end{cases}\\
\label{str}s_iT_j &=\begin{cases} T_2T_i^{-1} s_i, & i=j=1,3,4,\\
T_js_i, & i=1,3,4,\ j=1,2,3,4,\ i\neq j.
\end{cases}\\
\label{s2tr}s_2T_j&=\begin{cases}
T_js_2, & j=1,3,4,\\
 T_1T_2^{-1}T_3T_4s_2, & j=2.\\
\end{cases}
% \\
% \label{rbtr}r_iT_j&=\begin{cases} T_i^{-1} r_i, & i=j=1,3,4,\\
% T_kT_i^{-1}r_i, & \{i,j,k\}=\{1,3,4\}\\
% T_2T_i^{-2}r_i, & i=1,3,4,\ j=2.
%\end{cases}
\end{align} 
\end{subequations}

% The operators $T_k$ 
% coincide with the operators $T_{\varpi_k}$ of \cite{ny}, and are related to 
% the operators of \cite{m} by
% \begin{align*}T_{03}&=T_2T_3^{-1},& T_{14}&=T_1^{-1}T_2T_4^{-1},&
% \hat T_{34}&=T_3^{-1}T_4,& T_{34}&=T_2T_3^{-1}T_4^{-1}, \\
% T_1&=T_{03}T_{14}^{-1}\hat T_{34}^{-1},& T_2&=T_{03}^2T_{34}^{-1}\hat T_{34}^{-1},& T_3&=T_{03}T_{34}^{-1}\hat T_{34}^{-1},& T_4&=T_{03}T_{34}^{-1}. \end{align*}

% \vspace{1ex}
% \begin{center}
% \begin{tabular}{|c||ccccc|}
% \hline
%  & $\alpha_0$ & $\alpha_1$ & $\alpha_2$ & $\alpha_3$ & $\alpha_4$\\
% \hline \hline
% $T_1$ & $\alpha_0-1$ & $\alpha_1+1$ & $\alpha_2$ & $\alpha_3$ & $\alpha_4$\\
% $T_2$ & $\alpha_0-2$ & $\alpha_1$ & $\alpha_2+1$ & $\alpha_3$ & $\alpha_4$\\
% $T_3$ & $\alpha_0-1$ & $\alpha_1$ & $\alpha_2$ & $\alpha_3+1$ & $\alpha_4$\\
% $T_4$ & $\alpha_0-1$ & $\alpha_1$ & $\alpha_2$ & $\alpha_3$ & $\alpha_4+1$\\
% \hline
% \end{tabular}
% \end{center}
% \vspace{1ex}
% (This table is the opposite of \cite[Eq.\ (34)]{ny}, since we 
% consider B\"acklund transformations as acting on solutions
% rather than on the dependent variables.) 

\subsection{Tau functions}
\label{tss}

If one computes the action of some element of $\mathbb Z^4$ on the generator $q$, corresponding to a solution, one finds that it always factors. For instance,
\begin{align}\notag T_3(q)
&=\frac{t\big(p(t-q)+\alpha_0\big)}{qp(t-q)-\alpha_2q+(\alpha_0+\alpha_2)t}\\
&\quad\times\frac{qp(t-q)+(\alpha_0+\alpha_4)q-\alpha_4t}{qp(t-q)-(\alpha_1+\alpha_2)q+(\alpha_0+\alpha_1+\alpha_2)t}\label{t3q}. \end{align}
The non-trivial factors are essentially 
 \emph{tau functions}. To incorporate these, we need to work in a field extension
of  $\mathcal F_0$.  
One way to do this was proposed by Masuda \cite{m}. It is, however, not ideal for our purposes and we will therefore work with a variation of Masuda's construction.

We introduce the modified Hamiltonian
\begin{align}\notag h_0&=H+\frac t{12}\left(
2(\alpha_0-1)^2-\alpha_1^2+2\alpha_3^2-\alpha_4^2+6(\alpha_0-1)\alpha_3\right)\\
\label{mh}&\quad+\frac {t-1}{12}\left(2(\alpha_0-1)^2-\alpha_1^2-\alpha_3^2+2\alpha_4^2+6(\alpha_0-1)\alpha_4\right).
\end{align}
Note that \eqref{ph} holds with $H$ replaced by $h_0$. 
The extra terms have been introduced so that
\begin{align}\label{hmia} s_1(h_0)&=s_2(h_0)=s_3(h_0)=s_4(h_0)=h_0,\\ 
\label{hmib}t_1(h_0)&=\frac{h_0}{t}, \qquad t_3(h_0)=-h_0. \end{align}
(Masuda \cite{m}  works with a different modification that 
satisfies \eqref{hmia} but not \eqref{hmib}.)
We also define
$$h_1=r_1(h_0),\qquad h_3=r_3(h_0),\qquad h_4=r_4(h_0),\qquad h_2=h_1+s_1(h_1)-\frac t 3+\frac 16. $$

We  denote by $\mathcal F$ the field extension of 
  $\mathcal F_0$ by the additional generators $u$, $v$, $\tau_0,\dots,\tau_4$. The generators $u$ and $v$ satisfy
\begin{equation}\label{uvd}t=u^2v^4,\qquad 1-t=u^4v^2 \end{equation}
and thus formally correspond to the
roots $t^{-1/6}(1-t)^{1/3}$ and $t^{1/3}(1-t)^{-1/6}$.
We extend $\delta$ to the new generators by
$$\delta(u)=\frac{u(t+1)}{6},\qquad \delta(v)=\frac{v(t-2)}{6},\qquad \delta(\tau_j)=\tau_jh_j,\qquad j=0,\dots,4, $$
which is consistent with \eqref{uvd}.
Finally, we extend the action of the B\"acklund transformations by the 
following table.

\vspace{1ex}
\begin{center}
\begin{tabular}{|c||cc|c c c c c|}
\hline
 & $u$ & $v$ & $\tau_0$ & $\tau_1$ & $\tau_2$ & $\tau_3$ & $\tau_4$ \\
\hline \hline
&&&&&&&\\[-4mm]
$s_0$ & $u$ & $v$ & $\frac{\ti(t-q)\tau_2}{u^2v^2\tau_0} $ & $\tau_1$ & $\tau_2$ & $\tau_3$ & $\tau_4$\\
$s_1$ & $u$ & $v$ & $\tau_0$ & $\frac{\ti uv\tau_2}{\tau_1}$ & $\tau_2$ & $\tau_3$ & $\tau_4$\\
$s_2$ & $u$ & $v$ & $\tau_0$ & $\tau_1$ & $\frac{p\tau_0\tau_1\tau_3\tau_4}{\tau_2}$ & $\tau_3$ & $\tau_4$\\
$s_3$ & $u$ & $v$ & $\tau_0$ & $\tau_1$ & ${\tau_2}$ & $\frac{(1-q)\tau_2}{u\tau_3}$ & $\tau_4$\\
$s_4$ & $u$ & $v$ & $\tau_0$ & $\tau_1$ & ${\tau_2}$ & $\tau_3$ & $\frac{q\tau_2}{v\tau_4}$\\
$r_1$ & $u$ & $-v$ & $\tau_1$ & $\tau_0$ & $\frac{(q-t)\tau_2}{u^3v^3}$ & $\tau_4$ & $\tau_3$\\[1mm]
$r_3$ & $-u$ & $v$ & $\tau_3$ & $\tau_4$ & $\frac{\ti q\tau_2}{uv^2}$ & $\tau_0$ & $\tau_1$\\
$t_1$ & $u$ & $\frac{\ti}{uv}$ & $\ti\tau_0$ & $\tau_4$ & $-q\tau_2$ & $\tau_3$ & $\tau_1$\\
$t_3$ & $v$ & $u$ & $\ti\tau_0$ & $\tau_1$ & $\tau_2$ & $\tau_4$ & $\tau_3$\\
\hline
\end{tabular}
\end{center}
\vspace{1ex}

It is straight-forward to check that this is consistent with \eqref{uvd} and that  \eqref{bd} hold on the new generators.

The relations \eqref{wa}--\eqref{we} are all valid
on $\mathcal F$. In particular, the operators $T_j$ still define an action of
$\mathbb Z^4$. This is in contrast to the alternative definition of Masuda. 
However, the relations \eqref{wd}--\eqref{wi} are \emph{not} valid. 

The relations \eqref{wd} are replaced by
\begin{equation}\label{dcr}t_1^2=t_3^2=(t_1t_3)^3=\sigma,\qquad \sigma^2=1, 
\end{equation}
where $\sigma$ is the B\"acklund transformation mapping $\tau_0$ to $-\tau_0$ and fixing all other generators. Equivalently,
with $x=t_1$ and $a=t_1t_3$,
$$a^6=1,\qquad x^2=a^3,\qquad axa=x, $$ 
which are the standard defining relations for the dicyclic group $\operatorname{Dic}_3$ of order $12$. 

One has a lot of freedom when extending
 the B\"acklund transformations to $\mathcal F$. 
Our definition 
has been adapted to the  specific class of solutions that we will consider. 
For these solutions, \eqref{dcr} can be realized with $t_j$ represented by changes of  variable, see \eqref{xttx}. One can 
modify the definition of  $t_j$ so that the relations \eqref{wd} remain valid, but that would be inconvenient for our purposes.

One can write down modified versions of 
 \eqref{wh}--\eqref{wi} that are valid on $\mathcal F$, but we do not do so here. However, we will need the commutation relations between $t_j$ and the  embedded lattice $\mathbb Z^4$.

\begin{lemma}
Let $\psi_j$ and $\chi_j$ be the mutually commuting field automorphisms of $\mathcal F$ fixing $\mathcal F_0$ and acting on the remaining generators according to the following table.
\vspace{1ex}
\begin{center}
\begin{tabular}{|c||cc|c c c c c|}
\hline
 & $u$ & $v$ & $\tau_0$ & $\tau_1$ & $\tau_2$ & $\tau_3$ & $\tau_4$ \\
\hline \hline
$\psi_1$ & $-u$ & $-v$ & $-\ti\tau_0$ & $-\ti\tau_1$ & $\tau_2$ & $-\ti\tau_3$ & $-\ti\tau_4$\\
$\psi_2$ & $u$ & $v$ & $\tau_0$ & $-\tau_1$ & $\tau_2$ & $-\tau_3$ & $-\tau_4$\\
$\psi_3$ & $u$ & $-v$ & $-\ti\tau_0$ & $-\ti\tau_1$ & ${\tau_2}$ & $-\ti\tau_3$ & $-\ti\tau_4$\\
$\psi_4$ & $-u$ & $v$ & $-\ti\tau_0$ & $-\ti\tau_1$ & ${\tau_2}$ & $-\ti\tau_3$ & $-\ti{\tau_4}$\\
$\chi_1$& $u$ & $v$ & $\tau_0$ & $\tau_1$ & $\tau_2$ & $-\tau_3$ & $-\tau_4$ \\
$\chi_2$& $u$ & $v$ & $-\tau_0$ & $-\tau_1$ & $\tau_2$ & $-\tau_3$ & $-\tau_4$ \\
$\chi_3$& $u$ & $v$ & $\tau_0$ & $-\tau_1$ & $\tau_2$ & $\tau_3$ & $-\tau_4$ \\
$\chi_4$& $u$ & $v$ & $\tau_0$ & $-\tau_1$ & $\tau_2$ & $-\tau_3$ & $\tau_4$ \\
\hline
\end{tabular}
\end{center}
\vspace{1ex}
Then, the following relations hold on $\mathcal F$:
\begin{align}
\label{t1t}t_1T_1^{l_1}T_2^{l_2}T_3^{l_3}T_4^{l_4}&=T_1^{l_4}T_2^{l_2}T_3^{l_3}T_4^{l_1}t_1 X,\\
\label{t3t}t_3T_1^{l_1}T_2^{l_2}T_3^{l_3}T_4^{l_4}&=T_1^{l_1}T_2^{l_2}T_3^{l_4}T_4^{l_3}t_3 X,
\end{align}
where
 \begin{align}\notag X&=
 \psi_1^{l_1}\psi_2^{l_2}\psi_3^{l_3}\psi_4^{l_4}\chi_1^{\binom{l_1}2+l_3l_4}
 \chi_3^{\binom{l_3}2+l_1l_4}\chi_4^{\binom{l_4}2+l_1l_3}\\
\label{x} &\quad\times\chi_2^{l_2(l_1+l_3+l_4)+\frac 12(l_3-l_1)(l_4-l_1)(l_4-l_3)}.
 \end{align}
\end{lemma}

\begin{proof}
We  prove \eqref{t1t} by induction on $\sum_j|l_j|$, using the
 relations
\begin{align}
\label{stbt}t_1T_{1,2,3,4}&=T_{4,2,3,1}t_1\psi_{1,2,3,4},\\
\label{pst}\psi_jT_k&=T_k\psi_j\cdot\begin{cases}1, & j=k=2,\\
\chi_l, & j=k=l\neq 2\ \text{or}\  \{j,k,l\}=\{1,3,4\},\\
\chi_2, & \text{if exactly one of $j$ and $k$ equals $2$},
\end{cases} \\
\label{cht}\chi_jT_k&=T_k\chi_j\cdot\begin{cases}1, & j=k\ \text{or}\ 2\in\{j,k\},\\
\chi_2, & j\neq k\  \text{and}\  2\notin\{j,k\}.
\end{cases} \end{align}

Assuming \eqref{t1t}, we need to show that the same relation holds when  $l_j$ is replaced by $l_j\pm 1$ for some $j$. For instance, when $l_1$ is replaced by $l_1\pm 1$,
 the induction step would follow from
$$T_4t_1X_{l_1+1}=t_1X_{l_1}T_1, $$
where we indicate the dependence of $X$ on $l_1$. Using
\eqref{stbt}, this is reduced to
\begin{equation}\label{is}\psi_1^{-1}X_{l_1+1}=T_1^{-1}X_{l_1}T_1. \end{equation}
On the right-hand side of \eqref{is}, 
we apply conjugation by $T_1$ to each factor in
\eqref{x}. Using \eqref{pst} and \eqref{cht}, we obtain
\begin{multline*}(\psi_1\chi_1)^{l_1}(\psi_2\chi_2)^{l_2}(\psi_3\chi_4)^{l_3}(\psi_4\chi_3)^{l_4}\chi_1^{\binom{l_1}2+l_3l_4}
(\chi_2\chi_3)^{\binom{l_3}2+l_1l_4}(\chi_2\chi_4)^{\binom{l_4}2+l_1l_3}\\
\times\chi_2^{l_2(l_1+l_3+l_4)+\frac 12(l_3-l_1)(l_4-l_1)(l_4-l_3)}.\end{multline*}
Writing
$$\frac{(l_3-l_1)(l_4-l_1)(l_4-l_3)}2=\binom{l_1}2(l_4-l_3)+\binom{l_3}2(l_1-l_4)+\binom{l_4}2(l_3-l_1), $$
this is  seen to equal  the left-hand side of \eqref{is}.
More generally, the induction step in $l_j$ can be reduced to
$\psi_j^{-1}X_{l_j+1}=T_j^{-1}X_{l_j}T_j$, which for $j=3$ and $4$ follows from
\eqref{is} by symmetry, and for $j=2$ is proved similarly.
 Moreover, since
$$ t_3T_{1,2,3,4}=T_{1,2,4,3}t_3\psi_{1,2,3,4},$$
 the same arguments prove \eqref{t3t}. 
\end{proof}

We are interested in the lattice of tau functions
\begin{equation}\label{tf}\tau_{l_1l_2l_3l_4}=T_1^{l_1}T_2^{l_2}T_3^{l_3}T_4^{l_4}\tau_0.\end{equation}

\begin{corollary}\label{ttc}
The transformations $t_j$ act on \eqref{tf}
by
\begin{align*}t_1(\tau_{l_1l_2l_3l_4})&=i^{1+(2l_2-1)(l_1+l_3+l_4)+(l_3-l_1)(l_4-l_1)(l_4-l_3)}\tau_{l_4l_2l_3l_1},\\
 t_3(\tau_{l_1l_2l_3l_4})&=\ti^{1+(2l_2-1)(l_1+l_3+l_4)+(l_3-l_1)(l_4-l_1)(l_4-l_3)}\tau_{l_1l_2l_4l_3}
\end{align*}
\end{corollary} 

It will be convenient to introduce an index $l_0$, determined from the other $l_j$ by
\begin{equation}\label{alr}l_0+l_1+2l_2+l_3+l_4=0. \end{equation}

\begin{lemma}\label{lap} With $\mathbf T=T_1^{l_1}T_2^{l_2}T_3^{l_3}T_4^{l_4}$, we have
\begin{align}
\notag \mathbf T(\alpha_j)&=\alpha_j-l_j,\qquad j=0,\dots,4,\\
\notag \mathbf T(u)&=(-1)^{l_3+l_4}u,\\
\notag \mathbf T(v)&=(-1)^{l_1+l_4}v,\\
\label{tlq} \mathbf T(q)&=(-1)^{l_3+l_4}\frac{\ti uv^2\tau_{l_1,l_2,l_3,l_4+1}\tau_{l_1,l_2+1,l_3,l_4-1}}{\tau_{l_1+1,l_2,l_3,l_4}\tau_{l_1-1,l_2+1,l_3,l_4}},\\
\notag\mathbf T(p)&=-\frac{\tau_{l_1+1,l_2,l_3,l_4}\tau_{l_1-1,l_2+1.l_3,l_4}\tau_{l_1,l_2-1,l_3+1,l_4+1}}{u^2v^2\tau_{l_1,l_2,l_3,l_4}\tau_{l_1,l_2,l_3+1,l_4}\tau_{l_1,l_2,l_3,l_4+1}}.
\end{align}
\end{lemma}

This can be checked by  direct computation. To prove the identities
for $\mathbf T(q)$ and $\mathbf T(p)$ one first verifies them for $\mathbf T=\id$ and then act on both sides with $\mathbf T$.

As an example, when $\mathbf T=T_3$, Lemma \ref{lap} gives
$$T_3(q)=
-\frac{\ti u v^2 \tau_{0,0,1,1}\tau_{0,1,1,-1}}{\tau_{1,0,1,0}\tau_{-1,1,1,0}}, $$
where we can compute
\begin{align*}
\tau_{0,0,1,1}&=-\frac{\tau_3\tau_4}{uv\tau_0}\,\big(p(t-q)+\alpha_0\big),\\
\tau_{0,1,1,-1}&=-\frac{\tau_2\tau_3}{uv^2\tau_0\tau_4}\,\big(pq(t-q)+(\alpha_0+\alpha_4)q-\alpha_4t\big),\\
\tau_{1,0,1,0}&=\frac{\ti\tau_1\tau_3v}{t\tau_0}\,\big(pq(t-q)-\alpha_2q+(\alpha_0+\alpha_2)t\big),\\
\tau_{-1,1,1,0}&=-\frac{\tau_2\tau_3}{uv^2\tau_0\tau_1}\,\big(pq(t-q)-(\alpha_0+\alpha_2)q+(\alpha_0+\alpha_1+\alpha_2)t\big).
\end{align*}
Thus, we recover the factorization \eqref{t3q}.

\subsection{Bilinear identities}

Tau functions of Painlev\'e VI  satisfy many bilinear identities \cite{m}. We do not give an exhaustive list, but state a few examples that we need.

\begin{proposition}
The tau functions \eqref{tf} satisfy the bilinear identities
\begin{align}\notag &(
l_0+l_2+l_3-
\alpha_0-\alpha_2-\alpha_3)\,{\delta(\tau_{l_1,l_2,l_3,l_4})}\tau_{l_1,l_2+1,l_3-1,l_4}\\
\notag &\qquad\quad+(l_1+l_2+l_4-\alpha_1-\alpha_2-\alpha_4)\,\tau_{l_1,l_2,l_3,l_4}{\delta(\tau_{l_1,l_2+1,l_3-1,l_4})}\\
\notag &\qquad\quad+Q(l_0-\alpha_0,l_1-\alpha_1,l_3-\alpha_3,l_4-\alpha_4)\,{\tau_{l_1,l_2,l_3,l_4}}\tau_{l_1,l_2+1,l_3-1,l_4}\\
\label{fkr} &\qquad=u^2v^2\,{\tau_{l_1+1,l_2,l_3-1,l_4+1}\tau_{l_1-1,l_2+1,l_3,l_4-1}},\\
\notag &(l_0+l_2+l_3-\alpha_0-\alpha_2-\alpha_3)\,{\delta(\tau_{l_1,l_2,l_3,l_4})}\tau_{l_1+1,l_2-1,l_3,l_4+1}\\
\notag &\qquad\quad+(l_1+l_2+l_4-\alpha_1-\alpha_2-\alpha_4)\,\tau_{l_1,l_2,l_3,l_4}{\delta(\tau_{l_1+1,l_2-1,l_3,l_4+1})}\\
\notag &\qquad\quad+R(l_0-\alpha_0,l_1-\alpha_1,l_3-\alpha_3,l_4-\alpha_4)\,{\tau_{l_1,l_2,l_3,l_4}}\tau_{l_1+1,l_2-1,l_3,l_4+1}\\
\label{slr}& \qquad=u^2v^2\,{\tau_{l_1,l_2-1,l_3+1,l_4}\tau_{l_1+1,l_2,l_3-1,l_4+1}},\\
\notag &
\frac{\delta^2(\tau_{l_1l_2l_3l_4})\tau_{l_1l_2l_3l_4}}t-\frac{\delta(\tau_{l_1l_2l_3l_4})^2}{t}
-\delta(\tau_{l_1l_2l_3l_4})\tau_{l_1l_2l_3l_4}\\
\notag &\qquad\quad+S(l_0-\alpha_0,l_1-\alpha_1,l_3-\alpha_3,l_4-\alpha_4)\,\tau_{l_1l_2l_3l_4}^2\\
\label{tbr}& \qquad=(-1)^{l_3+l_4}\ti u\, \tau_{l_1,l_2+1,l_3-1,l_4}\tau_{l_1,l_2-1,l_3+1,l_4},
\end{align}
where
{\allowdisplaybreaks
\begin{align*}Q(l_0,l_1,l_3,l_4)&=\frac 1{12}(l_1-l_4)(l_0-l_3+1)(t-2)\\
&\quad+\frac 1{ 12}
\left((l_1+l_4)^2+\left(l_0-l_1+\frac 12\right)\left(l_3-l_4-\frac 12\right)-\frac 14\right)(t+1),\\
R(l_0,l_1,l_3,l_4)&=\frac 1{12}(l_1-l_4)(l_0-l_3+1)(t-2)\\
&\quad+\frac 1{ 12}
\left((l_0+l_3+1)^2+\left(l_0-l_1+\frac 12\right)\left(l_3-l_4-\frac 12\right)-\frac 14\right)(t+1),\\
S(l_0,l_1,l_3,l_4)&=\frac 1{12}\left(2l_0^2-l_1^2+2l_3^2-l_4^2+6l_0l_3+4l_0+6l_3+2\right). \end{align*}
}
\end{proposition}

\begin{proof}
It is enough to prove these identities when $l_j=0$ for all $j$, since
the general case  follows  using Lemma \ref{lap}. In that case, \eqref{fkr} takes the form
\begin{multline}\label{chr}-(\alpha_0+\alpha_2+\alpha_3)(h_0)-(\alpha_1+\alpha_2+\alpha_4)(T_2T_3^{-1})(h_0)+Q(-\alpha_0,-\alpha_1,-\alpha_3,-\alpha_4)\\
=u^2v^2\frac{\tau_{1,0,-1,1}\tau_{-1,1,0,-1}}{\tau_{0,1,-1,0}\tau_{0,0,0,0}}.\end{multline}
This can be checked by a straight-forward computation, which simplifies if one notes that, by \eqref{we} and \eqref{hmia}, 
$(T_2T_3^{-1})(h_0)=(r_3s_0)(h_0)$.  Acting with $s_2$ on
\eqref{chr}, using \eqref{s2tr} and \eqref{hmia}, gives the case $l_j\equiv 0$ of \eqref{slr}.  The case $l_j\equiv 0$ of \eqref{tbr} can be written
$$\frac{\delta(h_0)}{t}-h_0+S(-\alpha_0,-\alpha_1,-\alpha_3,-\alpha_4)=\ti u\frac{\tau_{0,1,-1,0}\tau_{0,-1,1,0}}{\tau_{0,0,0,0}^2}, $$
which can again be verified directly.
\end{proof}

% \begin{multline}\label{chr}(\alpha_1+\alpha_2+\alpha_4)(T_3^{-1}T_2)(h_0)-(\alpha_0+\alpha_2+\alpha_3)h_0-R\\
% =,\end{multline}
% where
% \begin{multline*}Q=\frac 1{12}(\alpha_0-\alpha_3)(\alpha_1-\alpha_4+1)(t-2)\\
% +\frac 1{ 12}
% \left((\alpha_0+\alpha_3)^2+\left(\alpha_0-\alpha_1-\frac 12\right)\left(\alpha_3-\alpha_4+\frac 12\right)-\frac 14\right)(t+1).
% \end{multline*}
% The modified Hamiltonian \eqref{hm} satisfies
% \begin{multline}\label{chr}(\alpha_0+\alpha_2+\alpha_3)(T_1^{-1}T_2)(h_0)-(\alpha_1+\alpha_2+\alpha_4)T_4(h_0)-Q\\
% =-u^2v^2\frac{\tau_{-1,0,1,0}\tau_{0,1,-1,1}}{\tau_{-1,1,0,0}\tau_{0,0,0,1}},\end{multline}
% \begin{multline}\label{chr}(\alpha_1+\alpha_2+\alpha_4)(T_3^{-1}T_2)(h_0)-(\alpha_0+\alpha_2+\alpha_3)h_0-R\\
% =,\end{multline}
% where
% \begin{multline*}Q=\frac 1{12}(\alpha_0-\alpha_3)(\alpha_1-\alpha_4+1)(t-2)\\
% +\frac 1{ 12}
% \left((\alpha_0+\alpha_3)^2+\left(\alpha_0-\alpha_1-\frac 12\right)\left(\alpha_3-\alpha_4+\frac 12\right)-\frac 14\right)(t+1).
% \end{multline*}

\subsection{Differential equations for tau functions}
Painlev\'e VI can be reformulated as a differential equation for the Hamiltonian, known as the $\operatorname{E_{VI}}$ equation \cite{jm,o}. 
In terms of the parameters 
$$b_1=\frac{\alpha_3+\alpha_4}{2},\quad b_2=\frac{\alpha_4-\alpha_3}2,\quad
b_3=\frac{\alpha_0+\alpha_1-1}{2},\quad b_4=\frac{\alpha_0-\alpha_1-1}{2} $$
it takes the form
\begin{multline}\label{e6}\frac{dh}{dt}\left(t(t-1)\frac{d^2h}{dt^2}\right)^2+\left(\frac{dh}{dt}\left(2h-(2t-1)\frac{dh}{dt}\right)+b_1b_2b_3b_4\right)^2\\
=\prod_{k=1}^4\left(\frac{dh}{dt}+b_k^2\right), \end{multline}
where $h$ is related to \eqref{mh} by 
\begin{equation}\label{oh}h=h_0-\frac{C}{24}\,(2t-1), \end{equation}
with
\begin{equation}\label{c}C=(\alpha_0-1)^2+\alpha_1^2+\alpha_3^2+\alpha_4^2=2(b_1^2+b_2^2+b_3^2+b_4^2). \end{equation}

Expressing $h$ in terms of tau functions, \eqref{e6} takes a rather complicated form. To obtain a simpler identity, we first cancel the term $\prod_k b_k^2$ 
and the factor $dh/dt$
on both sides, then  differentiate in $t$ and finally cancel the factor $d^2h/dt^2$. This leads to the alternative differential equation
\begin{multline*}
t^2(t-1)^2\frac{d^3h}{dt^3}+t(t-1)(2t-1)\frac{d^2h}{dt^2}+6t(t-1)
\left(\frac{dh}{dt}\right)^2+4(1-2t)h\frac{dh}{dt}\\
-\left(\sum_{k=1}^4b_k^2\right)
\frac{dh}{dt}+2h^2+(1-2t)b_1b_2b_3b_4-\frac 12\sum_{1\leq j<k\leq 4}b_j^2b_k^2=0.
\end{multline*}
Substituting \eqref{oh} and writing $\delta=t(t-1)d/dt$ gives
\begin{multline}\label{ahd}
\delta^{3}(h_0)+2(1-2t)\delta^2(h_0)+6\delta(h_0)^2+4(1-2t)h_0\delta(h_0)\\
-
\frac{(C-6)t(t-1)+C-3}3\,\delta(h_0)+2t(t-1)h_0^2+\frac{Ct(t-1)(2t-1)}{6}\,h_0\\
-\frac{t(t-1)G}8=0,
\end{multline}
where
\begin{multline}\label{g}
G=(\alpha_4-\alpha_3)(\alpha_3+\alpha_4)(\alpha_0+\alpha_1-1)(\alpha_0-\alpha_1-1)t\\
+(\alpha_3-\alpha_1)(\alpha_3+\alpha_1)(\alpha_0+\alpha_4-1)(\alpha_0-\alpha_4-1).
\end{multline}
One may check directly that \eqref{ahd}
 holds as an identity in $\mathcal F_0$.

Substituting $h_0=\delta(\tau_0)/\tau_0$ in \eqref{ahd}, 
all terms with denominators $\tau_0^3$ or $\tau_0^4$ cancel.
After simplification, we find that $\tau=\tau_0$ 
satisfies the  equation
\begin{multline}\label{qd}
\delta^4(\tau)\tau-4\delta^3(\tau)\delta(\tau)+2(1-2t)\delta^3(\tau)\tau+3\delta^2(\tau)^2-2(1-2t)\delta^2(\tau)\delta(\tau)\\
-\frac{(C-6)t(t-1)+C-3}3\,\delta^2(\tau)\tau+\frac{C(t^2-t+1)-3}{3}\,\delta(\tau)^2\\
+\frac{Ct(t-1)(2t-1)}{6}\,\delta(\tau)\tau-\frac{t(t-1)G}8\,\tau^2=0.
\end{multline}
Acting  by $\mathbb Z^4$, we obtain the following result, which we
 have not found  in  the literature. 
Analogous results for other Painlev\'e equations are discussed in \cite{c}.

\begin{proposition}\label{qdp}
The tau functions $\tau=\tau_{l_1l_2l_3l_4}$ satisfy \eqref{qd}, 
with $C$ and $G$  obtained from \eqref{c} and \eqref{g} by replacing
each  $\alpha_j$ with $\alpha_j-l_j$.
\end{proposition}

\section{Seed solution}
\label{sss}

\subsection{An algebraic Picard solution}
When 
$$\alpha_0=\alpha_1=\alpha_3=\alpha_4=0,\qquad \alpha_{2}=\frac 12,$$
 Painlev\'e VI
can be solved explicitly in terms of Weierstrass's $\wp$-function. 
This was done by Picard already in 1889 \cite{pi}. 
The general solution is labelled by two complex parameters $\nu_1$, $\nu_2$; it is algebraic if $\nu_1,\, \nu_2\in\mathbb Q$ \cite{maz}. In \cite{bm2}, the solution with $(\nu_1,\nu_2)=(1,1/3)$ was expressed as
\begin{equation}\label{pa}q^4-4tq^3+6tq^2-4tq+t^2=0. \end{equation}
We are interested in the corresponding lattice of tau functions \eqref{tf}.

We will substitute
$$t(\tau)=\frac{\tha(-1;p^6)^4}{\tha(-p^3;p^6)^4},\qquad p=e^{\pi\ti\tau}, $$
where $\tau$ is in the upper half-plane
and
$$\theta(x;p)=\prod_{j=0}^\infty(1-p^jx)\left(1-\frac{p^{j+1}}x\right). $$
We claim that $t$ is a \emph{modular function} for $\Gamma_0(6,2)$, that is, it is meromorphic and invariant under the modular transformations
\begin{equation}\label{mob}\tau\mapsto \frac{a\tau+b}{c\tau+d},\qquad \left[\begin{matrix}a&b\\c&d\end{matrix}\right]\in \mathrm{SL}(2,\mathbb Z),
\end{equation} 
   such that $b\equiv 0\ \operatorname{mod}\ 2$
 and  $c\equiv 0\ \operatorname{mod}\ 6$. 
To see this is, we use 
\cite[Lemma 9.1]{r1} to write
\begin{equation}\label{tz}t=\frac{\zeta(\zeta+2)^3}{(1+2\zeta)^3},\end{equation}
where 
$$\zeta=\frac{\om^2\tha(-1;p^2)\tha(-p\om;p^2)}{\tha(-p;p^2)\tha(-\om;p^2)}, \qquad \om=e^{2\pi\ti/3}.$$
 As is explained in \cite[\S 2.9]{r2}, $\zeta$ is a 
 \emph{Hauptmodul} for
$\Gamma_0(6,2)$, which means that  it generates the corresponding field
of modular functions.
Note that the values 
\begin{equation}\label{sc}\zeta=0,\ -1,\ 1,\ -2,\ -\frac 12,\ \infty\end{equation}
at the six cusps  of $\Gamma_0(6,2)$  correspond precisely to the three singular points $t=0,1,\infty$ of \eqref{py}.

Making the change of variables \eqref{tz} in \eqref{pa},
 we find that there is a rational solution in $\zeta$ (and thus modular in $\tau$) given by 
$$q=\frac{\zeta(\zeta+2)}{2\zeta+1}.$$
It is then easy to solve \eqref{ph}, giving
$$p=\frac{2\zeta+1}{2(1-\zeta)(\zeta+2)}.$$

We will write $\delta$ for the differentiation
\begin{equation}\label{dz}\delta=t(t-1)\frac{d}{dt}=\frac{\zeta(\zeta+1)(\zeta-1)(\zeta+2)}{2(2\zeta+1)^2}\frac{d}{d\zeta} \end{equation}
acting on $\mathbb C(\zeta)$.  
The fact that \eqref{pa} solves \eqref{py} 
can then be formulated  as 
\begin{equation}\label{xid}\mathbf X\circ\delta=\delta\circ\mathbf X, \end{equation}
where $\mathbf X:\ \mathcal F_0\rightarrow \mathbb C(\zeta)$ is the field automorphism
defined by
\begin{align*}{\bf X}(\alpha_0)&={\bf X}(\alpha_1)={\bf X}(\alpha_3)={\bf X}(\alpha_4)=0, \qquad {\bf X}(\alpha_2)=\frac 12, \\
{\bf X}(q)&=\frac{\zeta(\zeta+2)}{2\zeta+1}, \qquad
{\bf X}(p)=\frac{2\zeta+1}{2(1-\zeta)(\zeta+2)}, \qquad {\bf X}(t)=\frac{\zeta(\zeta+2)^3}{(1+2\zeta)^3}.\end{align*}

We remark that, in terms of the alternative 
 Hauptmodul $\bar\gamma=(1-\zeta)/(1+\zeta)$,
$$ t=\frac{(1-\bar\gamma)(3+\bar\gamma)^3}{(1+\bar\gamma)(3-\bar\gamma)^3}, \qquad q=\frac{(1-\bar\gamma)(3+\bar\gamma)}{(1+\bar\gamma)(3-\bar\gamma)}, $$
which  is the parametrization of \eqref{pa} used in \cite{bm2}.
Moreover, with $s=-\bar\gamma$ we have
$$\mathbf X\big(T_2^{-1}T_3^2r_4(q) \big)=\frac{3(3-s)(s+1)(s^2-3)^2}{(3+s)^2(s^6+3s^4-9s^2+9)}, $$
which is the solution \cite[(E.37)]{du}. 

\subsection{Modular tau functions} 
\label{mss}

Although our seed
 solution is modular for $\Gamma_0(6,2)$,
 the corresponding tau functions are only modular for a  subgroup. We will now describe the field $\mathcal M$ generated by all
the modular functions that we  need.

For  $p=e^{\pi\ti\tau}$, 
 Dedekind's eta function is given by
$$\eta(\tau)=p^{\frac{1}{12}}\prod_{j=0}^\infty(1-p^{2(j+1)}). $$
It satisfies
\begin{align}\label{dep}\eta(\tau+1)&=e^{\frac{\pi\ti}{12}}\eta(\tau),& \eta\left(-\frac 1\tau\right)&=\sqrt{-\ti \tau}\,\eta(\tau),&
 \eta\left(\tau+\frac 12\right)&=\frac{e^{\frac{\pi\ti}{24}}\eta(2\tau)^3}{\eta(\tau)\eta(4\tau)} .\end{align}
In the  notation
$$\left[a_1^{k_1},\dots,a_m^{k_m}\right](\tau)=\eta(a_1\tau)^{k_1}\dotsm\eta(a_m\tau)^{k_m},$$
we define $\mathcal M$ to be the field generated by
the five functions
$$\phi_1=\frac{[(1/2)^2]}{[1^2]},\quad \phi_2=\frac{[2^2]}{[1^2]},\quad 
\phi_3=\frac{[3/2]}{[1/2]},\quad \phi_4=\frac{[6]}{[2]},\quad \phi_{5}=\frac{[3]}{[1]}. $$
It can be deduced from \cite{gs}, or checked directly using Corollary \ref{cc2},
that the functions $\phi_j$ are all invariant under the  group
$K\subseteq \Gamma_0(6,2)$, consisting of transformations \eqref{mob} with
$$a\equiv d\equiv\pm 1\ \operatorname{mod}\ 12,\qquad 
b\equiv 0\ \operatorname{mod}\ 24,\qquad c\equiv 0\ \operatorname{mod}\ 72.$$
Thus, $\mathcal M$ is a field of modular functions for $K$; we do not know if it is in fact the field of all such functions.

The following identities
 can be proved by straight-forward manipulation of infinite products.

\begin{lemma}\label{tvl}
With $\om=e^{2\pi\ti/3}$,
\begin{align*}
\theta(p;p^2)&=p^{\frac 1{12}}{\phi_1}, &
\theta(-1;p^2)&=2p^{-\frac 1{6}}\phi_2,&
 \theta(-p;p^2)&=p^{\frac 1{12}}\frac{1}{\phi_1\phi_2},\\
\theta(p\omega;p^2)&=p^{\frac 1{12}}\frac{\phi_3}{\phi_5},&
\theta(-\omega;p^2)&=-\om^2p^{-\frac 1{6}}\frac{\phi_4}{\phi_5},&
\theta(\omega;p^2)&=(1-\omega)p^{-\frac 1{6}}\phi_5,
\\
\theta(-p\omega;p^2)&=p^{\frac 1{12}}\frac{\phi_5^2}{\phi_3\phi_4}.
\end{align*}
\end{lemma}

The normalizer of $K$ in $\mathrm{PSL}(2,\mathbb Z)$ is 
$\Gamma_0(3)$, consisting of transformations  \eqref{mob}
with $c\equiv 0\ \operatorname{mod}\ 3$. 
It is generated by  $T(\tau)=\tau+1$  and $U(\tau)=(\tau-1)/(3 \tau-2)$, subject to the single relation $U^3=1$ \cite{rad}.

\begin{lemma}\label{mal}
The generators of $\Gamma_0(3)$ act on the functions $[k/2]$, $k\mid 12$, according to
the following table, where $X=\sqrt{-\ti(3\tau-2)}$.
\begin{center}
\begin{tabular}{|c||c c c c c c|}
\hline
 & $[1/2]$ & $[1]$ & $[3/2]$ & $[2]$ & $[3]$ & $[6]$  \\
\hline \hline
&&&&&&\\[-4mm]
$T$& $e^{\frac{\pi\ti}{24}}\frac{[1^3]}{[1/2,2]}$ & $e^{\frac{\pi\ti}{12}}[1]$ & $e^{\frac{\pi\ti}8}\frac{[3^3]}{[3/2,6]}$ & $e^{\frac{\pi\ti}6}[2]$ & $e^{\frac{\pi\ti}4}[3]$ & $e^{\frac  {\pi\ti}2}[6]$\\[1mm]
$U$& $e^{-\frac{7\pi\ti}{24}}\frac{X[1^3]}{[1/2,2]}$ & $e^{-\frac{\pi\ti}{12}}X[1]$ & $e^{-\frac{\pi\ti}{24}}\frac{X[3^3]}{[3/2,6]}$ & $e^{\frac{\pi\ti}{12}}\frac{X}{\sqrt 2}[1/2]$ & $e^{-\frac{\pi\ti}{12}}X[3]$ & $e^{\frac{\pi\ti}{12}}\frac{X}{\sqrt 2}[3/2]$\\[1mm]
\hline
\end{tabular}
\end{center}
\end{lemma} 

\begin{proof}
This is straight-forward to verify using \eqref{dep}. 
For instance,
\begin{align*}[3/2](U\tau)&=\eta\left(\frac{3(\tau-1)}{2(3\tau-2)}\right)= 
\eta\left(\frac 12-\frac{1}{6\tau-4}\right)
=e^{\frac{\pi\ti}{24}}\frac{\eta\left(-\frac 1{3\tau-2}\right)^3}{\eta\left(-\frac 1{6\tau-4}\right)\eta\left(-\frac 2{3\tau-2}\right)}\\
&=e^{\frac{\pi\ti}{24}}X\frac{\eta(3\tau-2)^3}{\eta(6\tau-4)\eta\left(\frac{3\tau-2}2\right)}=e^{-\frac{\pi\ti}{24}}X\frac{\eta(3\tau)^3}{\eta(6\tau)\eta\left(\frac{3\tau}2\right)}.\end{align*}
\end{proof}

\begin{corollary}\label{cc2}
The group  $\Gamma_0(3)$ acts  on $\mathcal M$
according to the following table.
\begin{center}
\begin{tabular}{|c||c c c c c |}
\hline
 & $\phi_1$ & $\phi_2$ & $\phi_3$ & $\phi_4$ & $\phi_5$   \\
\hline \hline
&&&&&\\[-4mm]
$T$& $e^{-\frac{\pi\ti}{12}}\frac{1}{\phi_1\phi_2}$ & $e^{\frac{\pi\ti}6}\phi_2$ & $e^{\frac{\pi\ti}{12}}\frac{\phi_5^3}{\phi_3\phi_4}$ & $e^{\frac{\pi\ti}3}\phi_4$ & $e^{\frac{\pi\ti}6}\phi_5$ \\[1mm]
$U$& $e^{-\frac{5\pi\ti}{12}}\frac{1}{\phi_1\phi_2}$ & $e^{\frac{\pi\ti}3}\frac {\phi_1}{2}$ & $e^{\frac{\pi\ti}{4}}\frac{\phi_5^3}{\phi_3\phi_4}$ & $\phi_3$ & $\phi_4$\\[1mm]
\hline
\end{tabular}
\end{center}
\end{corollary}

One may check that the action described in Corollary \ref{cc2} 
factors to a faithful action of 
 $\Gamma_0(3)/K$, which is a group of order $12^3$. 
We will only work with the  subgroup of $\Gamma_0(3)$ generated by
$t_1=(UT^3)^3$ and $t_3=(T^3U)^3$. 
We use the same notation as for  B\"acklund transformations
in view of Proposition \ref{tep} below.

\begin{corollary}\label{cc3} The transformations $t_1$ and $t_3$ 
act on $\mathcal M$ according to the following table, where we  also 
introduce an auxiliary field automorphism $\sigma$.
\begin{center}
\begin{tabular}{|c|| c c c c c|}
\hline
 &  $\phi_1$ & $\phi_{2}$ & $\phi_3$ & $\phi_4$ & $\phi_5$  \\
\hline \hline
&&&&&\\[-4mm]
$t_1$& $\phi_1$ & $e^{\frac{3\pi\ti}{4}}\frac{1}{2\phi_1\phi_2}$ & $e^{\frac{3\pi\ti}2}\phi_3$ & $e^{\frac{5\pi\ti}4}\frac{\phi_5^3}{\phi_{3}\phi_4}$ & $e^{\frac{3\pi\ti}2} \phi_5$ \\[1mm]
$t_3$& $2e^{\frac{\pi\ti}{2}}{\phi_2}$ & $e^{\frac{3\pi\ti}{2}}\frac{\phi_1}{ 2}$ & $-\phi_4$ & $\phi_3$ & $e^{\frac{3\pi\ti}{2}}\phi_5$ \\
$\sigma$& $\phi_1$ & $\phi_2$ & $-\phi_3$ & $-\phi_4$ & $-\phi_5$ \\
\hline
\end{tabular}
\end{center}
These automorphisms satisfy \eqref{dcr} and thus generate the group $\operatorname{Dic}_3$.
\end{corollary}

 Combining \cite[Lemma 9.1]{r1}
and Lemma \ref{tvl} gives the following relations between
the generators of $\mathcal M$ and the Hauptmodul
$\zeta$.

\begin{lemma}\label{zpl}
We have
\begin{align*}
\zeta&=-2\frac{\phi_1\phi_2^2\phi_5^3}{\phi_3\phi_4^2},&
\zeta+1&=-\frac{\phi_1^2\phi_2\phi_5^3}{\phi_3^2\phi_4},&
\zeta-1&=-3\frac{\phi_1^2\phi_2\phi_3^2\phi_5}{\phi_4},\\
\zeta+2&=6\frac{\phi_1\phi_2^2\phi_4^2\phi_5}{\phi_3},&
2\zeta+1&=-3\frac{\phi_5^{10}}{\phi_3^4\phi_4^4}.
\end{align*}
In particular, the field of modular functions $\mathbb C(\zeta)$ for $\Gamma_0(6,2)$
is a subfield of $\mathcal M$. Moreover,
\begin{align*}\phi_1^{12}&=\frac{2^4(\zeta-1)^2(\zeta+1)^6}{\zeta^3(\zeta+2)(2\zeta+1)},& \phi_2^{12}&=-\frac{\zeta^6(\zeta+2)^2}{2^8(\zeta+1)^3(\zeta-1)(2\zeta+1)},\\
\phi_3^{12}&=\frac{(\zeta-1)^4(\zeta+2)(2\zeta+1)}{3^6\zeta(\zeta+1)^4},&
\phi_4^{12}&=-\frac{(\zeta-1)(\zeta+2)^4(2\zeta+1)}{3^6\zeta^4(\zeta+1)},\\
\phi_5^6&=\frac{(\zeta-1)(\zeta+2)(2\zeta+1)}{3^3\zeta(\zeta+1)}.
\end{align*}
\end{lemma}

It follows that
$$t_1(\zeta)=\frac 1{\zeta},\qquad t_3(\zeta)=-\zeta-1,\qquad \sigma(\zeta)=\zeta.$$
In agreement with \cite[\S 2.9]{r2},
these maps generate an action of $ 
\operatorname{Dic}_3/\{\sigma=1\}\simeq S_3$ on 
  $\mathbb C(\zeta)$.

\begin{corollary}\label{dl}

The field 
$\mathcal M$ is closed under the 
differentiation \eqref{dz}. 
\end{corollary}

\begin{proof}
By Lemma \ref{zpl}, $\phi_j^{12}\in\mathbb C(\zeta)$ for each $j$. 
Thus, 
$$\delta(\phi_j)=\frac{\phi_j}{12}\cdot\frac{\delta(\phi_j^{12})}{\phi_j^{12}}\in\mathcal M.
$$
\end{proof}

We are now ready to incorporate  tau functions in our seed 
solution.

\begin{proposition}\label{tep}
The following equations define an extension of 
$\mathbf X$ to a field automorphism $\mathcal F\rightarrow\mathcal M$:
\begin{align*}
{\bf X}(u)&=\frac{\phi_1^2\phi_3^4}{2^{2/3}\phi_5^4},& {\bf X}(v)&=-\frac{2^{4/3}\phi_2^2\phi_4^4}{\phi_5^4},\\
{\bf X}(\tau_0)&=\frac 1{\phi_5},&
{\bf X}(\tau_1)&=-\frac{\phi_3\phi_4}{\phi_5^2},&
{\bf X}(\tau_2)&=\frac{ 2^{-2/3}\ti\phi_5^4}{\phi_1^2\phi_2^2\phi_3^2\phi_4^2},\\
 {\bf X}(\tau_3)&=\frac{e^{\frac{\pi\ti}4}\phi_5}{\phi_3},
& {\bf X}(\tau_4)&=\frac{e^{\frac{3\pi\ti}4}\phi_5}{\phi_4}.
\end{align*}
This extension satisfies \eqref{xid}, as well as the identities
\begin{subequations}\label{xtsx}
\begin{align}\label{xttx}
\mathbf X\circ t_j&=t_j\circ\mathbf X,& j&=1,\,3,\\ 
\label{xss}\mathbf X\circ s_j&=\mathbf X,& j&=0,\,1,\,3,\,4.
\end{align}
\end{subequations}
\end{proposition}

It is useful to note that
\begin{subequations}\label{pmf}
\begin{align}
\mathbf X(u)&\simeq\frac{(\zeta+1)^{1/3}(\zeta-1)}{\zeta^{1/6}(\zeta+2)^{1/2}(2\zeta+1)^{1/2}},\\
\mathbf X(v)&\simeq\frac{\zeta^{1/3}(\zeta+2)}{(\zeta+1)^{1/6}(\zeta-1)^{1/2}(2\zeta+1)^{1/2}},\\
\mathbf X(\tau_0)&\simeq\frac{\zeta^{1/6}(\zeta+1)^{1/6}}{(2\zeta+1)^{1/6}(\zeta-1)^{1/6}(\zeta+2)^{1/6}},\\
\mathbf X(\tau_1)&\simeq\frac{(\zeta-1)^{1/12}(\zeta+2)^{1/12}}{(2\zeta+1)^{1/6}(\zeta+1)^{1/12}\zeta^{1/12}},\\
\label{xt2}\mathbf X(\tau_2)&\simeq\frac{(2\zeta+1)^{2/3}}{\zeta^{1/3}(\zeta+1)^{1/3}(\zeta-1)^{1/3}(\zeta+2)^{1/3}},\\
\mathbf X(\tau_3)&\simeq\frac{(2\zeta+1)^{1/12}(\zeta+1)^{1/6}(\zeta+2)^{1/12}}{\zeta^{1/12}(\zeta-1)^{1/6}},\\
\mathbf X(\tau_4)&\simeq\frac{(2\zeta+1)^{1/12}\zeta^{1/6}(\zeta-1)^{1/12}}{(\zeta+1)^{1/12}(\zeta+2)^{1/6}},
\end{align}
\end{subequations}
where  $f\simeq g^{1/N}$ (with $f,g\in\mathcal F$) means that $f^N/g\in\mathbb C$.
We will also need the identities
\begin{subequations}\label{sh}
{\allowdisplaybreaks
\begin{align}
\mathbf X\left(\frac{\delta(u)}{u}\right)&= \frac{\zeta^4+14\zeta^3+24\zeta^2+14\zeta+1}{6(2\zeta+1)^3},\\
\mathbf X\left(\frac{\delta(v)}{v}\right)&=\frac{\zeta^4-10\zeta^3-12\zeta^2-4\zeta-2}{6(2\zeta+1)^3},\\
\mathbf X\left(\frac{\delta(\tau_0)}{\tau_0}\right)&=-\frac{(\zeta^2+\zeta+1)^2}{6(2\zeta+1)^3},\\
\mathbf X\left(\frac{\delta(\tau_1)}{\tau_1}\right)&=-\frac{2\zeta^4+4\zeta^3-6\zeta^2-8\zeta-1}{12(2\zeta+1)^3},\\
\mathbf X\left(\frac{\delta(\tau_2)}{\tau_2}\right)&=-\frac{(\zeta^2+\zeta+1)(2\zeta^2+2\zeta-1)}{3(2\zeta+1)^3},\\
\mathbf X\left(\frac{\delta(\tau_3)}{\tau_3}\right)&=\frac{\zeta^4-4\zeta^3-12\zeta^2-4\zeta+1}{12(2\zeta+1)^3},\\
\mathbf X\left(\frac{\delta(\tau_4)}{\tau_4}\right)&=\frac{\zeta^4+8\zeta^3+6\zeta^2-4\zeta-2}{12(2\zeta+1)^3}.
\end{align}
}
\end{subequations}

% Note that, by Lemma \ref{tvl}, the tau functions (except for $\tau_2$) have very simple expressions in terms of theta values. This phenomenon seems to hold more generally in the Picard class, cf.\ \cite{ma}.

\begin{proof}[Proof of \emph{Proposition \ref{tep}}]
To prove that $\mathbf X$ is well-defined, we only need to check that it
respects the relations  \eqref{uvd}. This is clear since, by Lemma \ref{zpl},
$$t=16\frac{\phi_1^4\phi_2^8\phi_3^8\phi_4^{16}}{\phi_5^{24}},\qquad
1-t=-\frac{(\zeta+1)(\zeta-1)^3}{(2\zeta+1)^3}=\frac{\phi_1^8\phi_2^4\phi_3^{16}\phi_4^{8}}{\phi_5^{24}}. $$

We must also check that \eqref{xid} holds on each of the generators $u$,
$v$ and $\tau_j$. Let $g$ be  one of these generators. By 
\eqref{pmf},
 $\mathbf X(g^{12})\in\mathbb C(\zeta)$. Thus, we can rewrite the relevant identity as
$$\mathbf X\left(\frac{\delta(g)}{g}\right)=\frac{\delta(\mathbf X(g^{12}))}{12\mathbf X(g^{12})}, $$
which is  elementary to verify using 
\eqref{dz} and \eqref{sh}. 
The remaining statements are easy to verify.
\end{proof}

In particular, 
$\mathbf X(s_j(\tau_j)/\tau_j)=1$ for $j=0,1,3,4$,
that is,
\begin{multline}\label{xti}\frac{\ti(2\zeta+1)^3}{3\zeta(\zeta+1)(\zeta-1)(\zeta+2)}\,
\mathbf X\left(\frac{u^2v^2\tau_0^2}{\tau_2}\right)=
-\ti\,\mathbf X\left(\frac{\tau_1^2}{uv\tau_2}\right)\\
=
\frac{(2\zeta+1)}{(\zeta+1)(1-\zeta)}\,\mathbf X\left(\frac{u\tau_3^2}{\tau_2}\right)=
\frac{2\zeta+1}{\zeta(\zeta+2)}\,\mathbf X\left(\frac{v \tau_4^2}{\tau_2}\right)=1.
\end{multline}

\section{Identification of tau functions}
\label{its}
Let $k_0,\dots,k_3$ be integers such that  $\sum_j k_j$ is even and write
$$n=\frac{k_0+k_1+k_2+k_3}{2}.$$
As in \cite{r3}, we  write
 $t^{(k_0,k_1,k_2,k_3)}=T_{n}^{(k_0,k_1,k_2,k_3)}$, where
 $T_{n}^{(k_0,k_1,k_2,k_3)}$ was introduced in \cite{r2}. 
In general,  $T_{n}^{(k_0,k_1,k_2,k_3)}$ is a function of $m=2n-\sum_j k_j$ variables
and one parameter $\zeta$; in this paper we are only concerned with the 
case $m=0$.
Our main result is that, up to elementary factors,
the    tau function
$\mathbf X(\tau_{l_1,l_2.l_3,l_4})$ can be identified with  $t^{(k_0,k_1,k_2,k_3)}$, 
where
\begin{subequations}\label{klr}
\begin{align}\label{klra}
k_0&=l_0+l_2+l_3, & k_1&=-l_2,& k_2&=l_0+l_2+l_4, & k_3&=l_0+l_1+l_2,
\end{align}
with $l_0$  given by \eqref{alr}. Equivalently,
\begin{align}
l_0&=n, & l_1&=k_1+k_3-n, & l_2 &=-k_1,\\
l_3&=k_0+k_1-n, & l_4&=k_1+k_2-n.
\end{align}
\end{subequations}
 The map
$(l_1,l_2,l_3,l_4)\mapsto (k_0,k_1,k_2,k_3)$ is a bijection from
$\mathbb Z^4$ to the sublattice of $\mathbb Z^4$ defined by  
$k_0+k_1+k_2+k_3\in 2\mathbb Z$.

Note also that, by
\eqref{sbtr}, 
$$s_2s_0T_1^{l_1}T_2^{l_2}T_3^{l_3}T_4^{l_4}(s_2s_0)^{-1}=T_1^{\hat l_1}T_2^{\hat l_2}T_3^{\hat l_3}T_4^{\hat l_4}, $$
where $\hat l_j=l_j+l_0+l_2$ for $j\neq 2$ and $\hat l_2=l_1+l_2+l_3+l_4$.
We can then write \eqref{klra} more symmetrically as
$k_{0,1,2,3}=\hat l_2+\hat l_{3,0,4,1}.$

We summarize the main facts about the function $t^{(k_0,k_1,k_2,k_3)}$
as follows, see \cite[Cor.\ 3.9 and Thm.\ 4.1]{r2}, \cite[Thm.\ 4.1]{r3}.
For $a\in\mathbb C\cup\{\infty\}$, we 
write $\mathrm o_a(f)$ for the order of a meromorphic function $f$ at the point $a$, that is, $\lim_{z\rightarrow a}(z-a)^{-\mathrm o_a(f)}f(z)$  and
$\lim_{z\rightarrow \infty}z^{\mathrm o_\infty(f)}f(z)$ are 
 finite and non-zero. We will write (cf.\ \eqref{sc})
 \begin{equation}\label{fc}\Lambda=\{0,1,-1,-2,-1/2\}.\end{equation}

\begin{proposition}[\cite{r2,r3}]\label{trp}
The functions $t^{(\mathbf k)}=t^{(k_0,k_1,k_2,k_3)}(\zeta)$ 
is a rational function in $\zeta$ with no poles outside $\Lambda$.
At   $\zeta=0$ and $\zeta=-2$, it has order
\begin{subequations}\label{ot}
\begin{align}\label{ota}\mathrm o_0(t^{(\mathbf k)})&=(k_1+k_2)(2n-k_1-k_2-1)+\max\big((n+1)(k_1+k_2-n),0\big),\\
\label{otb}\mathrm o_{-2}(t^{(\mathbf k)})&=\left[\frac{(k_1+k_2-1)^2}4\right]-(k_1+k_2)(n-1).
\end{align}
\end{subequations}
Moreover, the functions $t^{(\mathbf k)}$  
are uniquely determined by
the two recursions
\begin{subequations}\label{tdp}
\begin{multline}\label{kma}
t^{(\mathbf k-2\mathbf e_0)}t^{(\mathbf k+\mathbf e_0+\mathbf e_1)}=\zeta^2(\zeta+1)(\zeta-1)(2\zeta+1)^2\\
\times\left(\frac{1}{2k_0-1}\,t^{(\mathbf k)}\frac{d t^{(\mathbf k-\mathbf e_0+\mathbf e_1)}}{d \zeta}-\frac{1}{2k_0+1}\frac{d t^{(\mathbf k)}}{d \zeta} t^{(\mathbf k-\mathbf e_0+\mathbf e_1)}\right)\\
+\frac{\zeta(2\zeta+1)}{2(2k_0-1)(2k_0+1)(\zeta+2)}\,A^{(\mathbf k)}t^{(\mathbf k)}t^{(\mathbf k-\mathbf e_0+\mathbf e_1)},
\end{multline}
\begin{multline}\label{kmb}
t^{(\mathbf k-2\mathbf e_0)}t^{(\mathbf k+\mathbf e_0-\mathbf e_1)}=\frac{(\zeta+1)(\zeta-1)(\zeta+2)^2(2\zeta+1)^2}{\zeta^2}\\
\times\left(\frac{1}{2k_0-1}\,t^{(\mathbf k)}\frac{d t^{(\mathbf k-\mathbf e_0-\mathbf e_1)}}{d \zeta}-\frac{1}{2k_0+1}\frac{d t^{(\mathbf k)}}{d \zeta} t^{(\mathbf k-\mathbf e_0-\mathbf e_1)}\right)\\
+\frac{(2\zeta+1)(\zeta+2)}{2(2k_0-1)(2k_0+1)\zeta^3}\,B^{(\mathbf k)}t^{(\mathbf k)}t^{(\mathbf k-\mathbf e_0-\mathbf e_1)},
\end{multline}
where $\mathbf e_j$ are unit vectors and $A^{(\mathbf k)}$ and $B^{(\mathbf k)}$ certain explicit polynomials (see \cite{r3}) in $\zeta$ and $k_j$, the three initial values
\begin{align}
t^{(0,0,0,0)}&=t^{(1,-1,0,0)}=1,\\
t^{(0,-1,-1,0)}&=-\frac{2\zeta^2(\zeta-1)(\zeta+1)^2(2\zeta+1)}{(\zeta+2)^2}
\end{align}
and the four symmetries 
\begin{align}\label{tsa}t^{(k_0,k_1,k_2,k_3)}(\zeta)&=(2\zeta+1)^{(k_0+k_2-n)(n-1)}\left(\frac{\zeta}{\zeta+2}\right)^{(k_1+k_2-n)(n-1)} t^{(k_1,k_0,k_3,k_2)}(\zeta)\\
\notag &=\left(-\frac{1}{12}\right)^{n+1}\frac{1}{\zeta^{2(k_1+k_2+2n+3)}(\zeta+1)^{2(k_0+k_1+2n+3)}(\zeta-1)^{2(k_2+k_3+1)}}\\
\label{tsb}&\quad\times\frac{(\zeta+2)^{2(k_1+k_2+n+2)}}{(2\zeta+1)^{2(k_0+k_2+1)}}\, t^{(-k_0-1,-k_1-1,-k_2-1,-k_3-1)}(\zeta)\\
\label{tsc}&=\left(\frac{\zeta^3(2\zeta+1)}{\zeta+2}\right)^{n(n-1)}t^{(k_0,k_1,k_3,k_2)}(\zeta^{-1})\\
\label{tsd}&=\left(\frac{\zeta-1}{\zeta+2}\right)^{n(n-1)}t^{(k_2,k_1,k_0,k_3)}(-\zeta-1).
 \end{align}
\end{subequations}
\end{proposition} 

We have only given the  order at the cusps corresponding to the singular point $t=0$ of \eqref{py}. The behaviour at the other cusps follows using  \eqref{tsc}--\eqref{tsd}.

Let us introduce  the normalizing factor
$\phi_{l_1l_2l_3l_4}\in\mathcal F$
given by
\begin{align*}\notag\phi_{l_1l_2l_3l_4}&=\frac{(-1)^{\binom{l_1+1}3+\binom{l_3+1}3+\binom{l_4+1}3+\left(\binom{l_3+1}2+l_1l_3+l_2\right)l_4}\ti^{\binom{l_3+1}2+\binom{l_4+1}2-l_1^2l_3+l_1l_4^2+l_2+l_3+l_4}}{2^{l_0(l_0-1)+l_1^2+l_3^2+l_4^2}}\\
\notag&\quad\times\zeta^{l_4^2-l_0(l_0-1)-(l_0+l_2)(l_2+l_4)}(\zeta+1)^{l_3^2-l_0(l_0-1)-(l_0+l_2)(l_2+l_3)}\\
\notag&\quad\times(\zeta-1)^{(l_0+l_2)(l_1+l_4)-(l_2+l_3)^2-l_3}(\zeta+2)^{-3l_2(l_0+l_2+l_4)-(l_0+l_4)(l_4+1)}\\
\notag&\quad\times(2\zeta+1)^{-l_0^2-l_1(l_0+l_1+3l_2+1)-l_2}u^{\frac 12(l_1-l_3)(l_1+l_3+2l_4-1)+2l_2(l_0+l_2)}\\
&\quad\times 
v^{\frac 12(l_1-l_4)(l_1+l_4+2l_3-1)+2l_2(l_0+l_2)}\tau_0^{l_0+1}\tau_1^{l_1}\tau_2^{l_2}\tau_3^{l_3}\tau_4^{l_4}.
\end{align*}
It is sometimes useful to  write, with notation as in \eqref{pmf},
\begin{align}
\notag\phi_{l_1l_2l_3l_4}&\simeq \zeta^{\frac 1{12}\left(-10l_0^2-l_1^2-l_3^2+14l_4^2-6l_0l_4+16l_0+6l_4+2\right)}\\
\notag&\quad\times(\zeta+1)^{\frac{1}{12}\left(-10l_0^2-l_1^2+14l_3^2-l_4^2-6l_0l_3+16l_0+6l_3+2\right)}\\
\notag&\quad\times(\zeta-1)^{\frac{1}{12}\left(-6l_0^2-3l_1^2-6l_3^2-3l_4^2+6l_0l_3-6l_3-2\right)}\\
\notag&\quad\times(\zeta+2)^{\frac{1}{12}\left(6l_0^2-3l_1^2-3l_3^2-6l_4^2+6l_0l_4-12l_0-6l_4-2\right)}\\
\label{pph}&\quad\times(2\zeta+1)^{\frac{1}{12}\left(-6l_0^2-6l_1^2-3l_3^2-3l_4^2+6l_0l_1-6l_1-2\right)}.
\end{align}
Moreover, one can compute
\begin{align}
\notag\frac{\delta(\phi_{l_1l_2l_3l_4})}{\phi_{l_1l_2l_3l_4}}&=\frac{1}{12(2\zeta+1)^3}
\Big(-(26\zeta^4+70\zeta^3+6\zeta^2-38\zeta-10)l_0^2\\
\notag&\quad-(14\zeta^4+28\zeta^3-14\zeta-1)l_1^2+(\zeta^4-10\zeta^3-36\zeta^2-10\zeta+1)l_3^2\\
\notag&\quad+(\zeta^4+14\zeta^3-28\zeta-14)l_4^2+6\zeta(\zeta+1)(\zeta-1)(\zeta+2)l_1(l_0-1)\\
\notag&\quad+6\zeta(\zeta+2)(2\zeta+1)l_3(l_0-1)-6(\zeta+1)(\zeta-1)(2\zeta+1)l_4(l_0-1)\\
\label{dpn}&\quad+2(\zeta-1)(2\zeta+1)(5\zeta^2+17\zeta+8)l_0-2(\zeta^2+\zeta+1)^2\Big).
\end{align}
Let
$(Y_k)_{k\in\mathbb Z}$ be the solution to the recursion
$$Y_{k+1}Y_{k-1}=2(2k+1)Y_k^2,\qquad Y_0=Y_1=1,$$
that is,
$$Y_k=\begin{cases}
\prod_{j=1}^k \frac{(2j-1)!}{(j-1)!}, & k\geq 0,\\[1mm]
\frac{(-1)^{\frac{k(k+1)}{2}}}{2^{2k+1}}\prod_{j=1}^{-k-1} \frac{(2j-1)!}{(j-1)!}, & k<0.
\end{cases} $$
We can then formulate our main result as follows.

\begin{theorem}\label{trt}
We have
\begin{equation}\label{xtt}\mathbf X\left(\frac{\tau_{l_1l_2l_3l_4}}{\phi_{l_1l_2l_3l_4}}\right)=Y_{k_0}Y_{k_1}Y_{k_2}Y_{k_3}\,t^{(k_0,k_1,k_2,k_3)}, \end{equation}
where $k_j$ and $l_j$ are related by \eqref{klr}.
\end{theorem}

After our preparations (which  take up a large portion also of the papers 
\cite{r2,r3}) the proof of Theorem \ref{trt} is straight-forward.

\begin{proof}
It is enough to show that, if we use \eqref{xtt} to define $t^{(\mathbf k)}$, all
  the  properties \eqref{tdp} are valid. We first  apply $\mathbf X$ to 
\eqref{slr} and  substitute \eqref{xtt}.
Using that
$$\mathbf X\left(u^2v^2\frac{\phi_{l_1,l_2-1,l_3+1,l_4}\phi_{l_1+1,l_2,l_3-1,l_4+1}}{\phi_{l_1,l_2,l_3,l_4}\phi_{l_1+1,l_2-1,l_3,l_4+1}}\right)=-\frac{\zeta+2}{16\zeta(2\zeta+1)^4} $$
and
$$\frac{Y_{k+1}Y_{k-2}}{Y_kY_{k-1}}=4(2k+1)(2k-1), $$
we find that \eqref{kma} holds  with
\begin{multline*}A^{(\mathbf k)}=-8(2\zeta+1)^3\left(R(l_0,l_1,l_3,l_4)+\left(k_0-\frac 12\right)\frac{\delta(\phi_{l_1l_2l_3l_4})}{\phi_{l_1l_2l_3l_4}}\right.\\
\left.-\left(k_0+\frac 12\right)\frac{\delta(\phi_{l_1+1,l_2-1,l_3,l_4+1})}{\phi_{l_1+1,l_2-1,l_3,l_4+1}}\right). \end{multline*}
Similarly, \eqref{fkr} yields \eqref{kmb}  with
\begin{multline*}B^{(\mathbf k)}=-8(2\zeta+1)^3\left(Q(l_0,l_1,l_3,l_4)+\left(k_0-\frac 12\right)\frac{\delta(\phi_{l_1l_2l_3l_4})}{\phi_{l_1l_2l_3l_4}}\right.\\
\left.-\left(k_0+\frac 12\right)\frac{\delta(\phi_{l_1,l_2+1,l_3-1,l_4})}{\phi_{l_1,l_2+1,l_3-1,l_4}}\right). \end{multline*}
Using \eqref{dpn}, one may check that this agrees with the 
explicit expressions for $A^{(\mathbf k)}$ and 
$B^{(\mathbf k)}$ given in \cite{r3}.

The initial values are trivial to check. 
To prove \eqref{tsa}, we apply the identity
$\mathbf X\circ s_1s_4=\mathbf X $, which follows from \eqref{xss},
 to
$\tau_{l_1l_2l_3l_4}$. Since, by \eqref{str},
$$s_1s_4T_1^{l_1}T_2^{l_2}T_3^{l_3}T_4^{l_4}=T_1^{-l_1}T_2^{l_1+l_2+l_4}T_3^{l_3}T_4^{-l_4}s_1s_4, $$
we obtain
$\mathbf X(\tau_{l_1l_2l_3l_4})=\mathbf X(\tau_{-l_1,l_1+l_2+l_4,l_3,-l_4})$. 
We substitute \eqref{xtt} and write
$$\frac{\phi_{l_1l_2l_3l_4}}{\phi_{-l_1,l_1+l_2+l_4,l_3,-l_4}}=\frac{(\zeta+2)^{l_0l_4-2l_4}(2\zeta+1)^{l_0l_1-l_1+l_4}}{\zeta^{l_0l_4}}
\left(\frac{\tau_1^2}{\ti uv\tau_2}\right)^{l_1}\left(\frac{v\tau_4^2}{\tau_2}\right)^{l_4}.$$
Applying \eqref{xti} yields \eqref{tsa} after simplification.

Similarly,
starting  from $\mathbf X\circ s_0s_1s_3s_4=\mathbf X$ and using  
$s_0(\tau_0)=T_2(\tau_0)$, we get
 $\mathbf X(\tau_{l_1l_2l_3l_4})=\mathbf X(\tau_{-l_1,1-l_2,-l_3,-l_4})$. 
Writing
\begin{multline*}
\frac{\phi_{l_1l_2l_3l_4}}{\phi_{-l_1,1-l_2,-l_3,-l_4}}=(-1)^{l_1l_3+l_1l_4+l_3l_4}
\left(\frac{2^6\zeta^5(\zeta+1)^5(\zeta-1)(2\zeta+1)^5u^2v^2\tau_0^2}{\ti(\zeta+2)^5\tau_2}\right)^{l_0+1}\\
\times\left(\frac{\tau_1^2}{\ti(2\zeta+1)^2uv\tau_2}\right)^{l_1}
\left(-\frac{(\zeta+1)(2\zeta+1)u\tau_3^2}{(\zeta-1)^3\tau_2}\right)^{l_3}
\left(\frac{\zeta(2\zeta+1)v\tau_4^2}{(\zeta+2)^3\tau_2}\right)^{l_4}
\end{multline*}
and noting that 
$$\prod_{j=0}^3\frac{Y_{k_j}}{Y_{-k_j-1}}=\prod_{j=0}^3\frac{(-1)^{\frac{k_j(k_j+1)}2}}{2^{2k_j+1}}=\frac{(-1)^{l_1l_3+l_1l_4+l_3l_4}}{16^{l_0+1}}$$
gives \eqref{tsb}.

Next, by  Corollary \ref{ttc} and \eqref{xttx},
\begin{equation}\label{txt}t_1 \mathbf X(\tau_{l_1l_2l_3l_4})=i^{1+(2l_2-1)(l_1+l_3+l_4)+(l_3-l_1)(l_4-l_1)(l_4-l_3)}\mathbf X(\tau_{l_4l_2l_3l_1}).\end{equation}
Then, \eqref{tsc} follows from the easily verified identity
\begin{multline*}
\frac{t_1(\phi_{l_1l_2l_3l_4})}{\phi_{l_4l_2l_3l_1}}=i^{1+(2l_2-1)(l_1+l_3+l_4)+(l_3-l_1)(l_4-l_1)(l_4-l_3)}\\
\times\left(\frac{\zeta^3(2\zeta+1)}{\zeta+2}\right)^{l_0(l_0-1)}\left(\frac{(2\zeta+1)q}{\zeta(\zeta+2)}\right)^{l_2}.
\end{multline*}
The last symmetry, \eqref{tsd}, is proved similarly.
\end{proof}

\section{Applications}
\label{aps}

\subsection{Behaviour at singular points}
As a first
 application of Theorem \ref{trt}, we can compute the leading behaviour of
the tau functions at the cusps. Let
\begin{align*}C_0(l_0,l_1,l_3,l_4)&=\frac{l_0^2}6-\frac{l_1^2}{12}-\frac{l_3^2}{12}+\frac{l_4^2}6-\frac{l_0l_4}2+\frac{l_0}3-\frac{l_4}2+\frac 16+\max\big((l_0+1)l_4,0\big),\\
C_{-2}(l_0,l_1,l_3,l_4)&=-\frac{l_0^2}2-\frac{l_1^2}4-\frac{l_3^2}4-\frac{l_4^2}2-\frac{l_0l_4}2+\frac{l_4}2-\frac 16+\left[\frac{(l_0+l_4-1)^2}4\right]
\end{align*}
and define
\begin{align*}
C_{-1}(l_0,l_1,l_3,l_4)&=C_0(l_0,l_1,l_4,l_3), & C_\infty(l_0,l_1,l_3,l_4)&=C_0(l_0,l_4,l_3,l_1),\\
C_1(l_0,l_1,l_3,l_4)&=C_{-2}(l_0,l_1,l_4,l_3),& C_{-1/2}(l_0,l_1,l_3,l_4)&=C_{-2}(l_0,l_4,l_3,l_1).
\end{align*}

\begin{corollary}\label{ctc}
With $\Lambda$ as in \eqref{fc} and using the notation $\simeq$ as in
\eqref{pmf}, 
\begin{equation}\label{xts}\mathbf X(\tau_{l_1l_2l_3l_4})\simeq\prod_{a\in\Lambda}(\zeta-a)^{C_a(l_0,l_1,l_3,l_4)}\,
 p(\zeta), \end{equation}
where $p$ is a polynomial of degree
\begin{multline*}\deg(p)=-\sum_{a\in\Lambda\cup\{\infty\}}C_a(l_0,l_1,l_3,l_4)\\
=2\binom{l_2}2+\sum_{j\in\{1,3,4\}}\left(2\binom{l_2+l_j+1}{2}+\max\big((l_0+1)l_j,0\big)+\left[\frac{(l_0+l_j-1)^2}4\right]\right),\end{multline*}
which does not vanish at $\Lambda$.
\end{corollary}

\begin{proof}
Combining \eqref{ot}, \eqref{pph}  and \eqref{xtt} yields the given expressions for $C_0$ and $C_{-2}$. By \eqref{txt} and the corresponding equation for $t_3$, it follows that $\mathrm o_a(\mathbf X(\tau_{l_1l_2l_3l_4}))=C_a(l_0,l_1,l_3,l_4)$ for each $a\in\Lambda\cup\infty$. This
proves \eqref{xts} and the first expression  for $\deg(p)$. The second expression follows by a direct computation.
\end{proof}

Having understood the behaviour of the tau functions at the cusps, it is easy
to understand the corresponding solutions. Let
$$q_{l_1l_2l_3l_4}=\mathbf X(T_1^{l_1}T_2^{l_2}T_3^{l_3}T_4^{l_4} q)\in\mathbb C(\zeta). $$
Recall that $q=q_{l_1l_2l_3l_4}$ solves
 \eqref{py}, with $t$  given by \eqref{tz}
and 
$$(\alpha,\beta,\gamma,\delta)=\left(\frac{l_1^2}2,-\frac{l_4^2}2,\frac{l_3^2}2,\frac{1-l_0^2}2\right).$$

\begin{corollary}\label{scc}
Define $\chi(k)$  as $1$ for $k$ odd and $0$ for $k$ even.
Then,
\begin{subequations}
\begin{align}\label{sca}q_{l_1l_2l_3l_4}&=\frac{\zeta^{1+|l_0|\delta_{l_4,0}}(\zeta+2)^{1+\chi(l_1+l_3)}}{(2\zeta+1)^{1+\chi(l_3+l_4)}}\,f(\zeta)\\
\label{scb}&=1+\frac{(\zeta+1)^{1+|l_0|\delta_{l_3,0}}(\zeta-1)^{1+\chi(l_1+l_4)}}{(2\zeta+1)^{1+\chi(l_3+l_4)}}\,g(\zeta),\end{align}
\end{subequations}
with $f$ and $g$ rational functions with no zeroes or poles in $\Lambda$.
Moreover,
\begin{equation}\label{oq}\mathrm o_\infty(q_{l_1l_2l_3l_4})=1+|l_0|\delta_{l_1,0}.
\end{equation}
\end{corollary}

\begin{proof}
If we substitute \eqref{xts} in \eqref{tlq} and simplify, using 
\begin{gather*}\max(k(l+1),0)+\max(k(l-1),0)-2\max(kl,0)=|k|\delta_{l,0},\\
\left[\frac{(l+1)^2}{4}\right]+\left[\frac{(l-1)^2}{4}\right]-2\left[\frac{l^2}{4}\right]=\chi(l),
\end{gather*} 
we obtain \eqref{sca} and \eqref{oq}. Applying $t_3$, using \eqref{t3t} and \eqref{xttx}, yields \eqref{scb}.
\end{proof}

Corollary \ref{scc} immediately gives the behaviour of the solutions near the singular points of \eqref{py}. For instance, near $t=0$, \eqref{tz} behaves either as $t\sim \zeta$ or as $t\sim(\zeta+2)^3$. The first branch corresponds to $q\sim t^{1+|l_0|\delta_{l_4,0}}$ and the second branch to $q\sim t^{1/3}$ or $q\sim t^{2/3}$, depending on the parity of $l_1+l_3$.
In the terminology of \cite[\S 2.9]{r2}, the first type of solution appears at the hyperbolic cusps and the second type at the trigonometric cusps. In the context of the XYZ model, these cusps corresponds to degenerations to the XY and XXZ model, respectively.

\subsection{Properties of the functions $t^{(\mathbf k)}$}

We can apply Theorem \ref{trt} to deduce new properties of
the functions  $t^{(\mathbf k)}$. For instance, we can obtain the following new symmetry. We do not know how to obtain this result
without using the relation to tau functions.

\begin{corollary}\label{nsc}
The functions $t^{(k_0,k_1,k_2,k_3)}$ satisfy
\begin{align}
\notag t^{(k_0,k_1,k_2,k_3)}(\zeta)&=(-1)^{(k_0+k_1+n)(k_1+k_3+n)}\frac{Y_{n-k_0}Y_{n-k_1}Y_{n-k_2}Y_{n-k_3}}{Y_{k_0}Y_{k_1}Y_{k_2}Y_{k_3}}\\
\notag&\quad\times\left(\frac{\zeta^{k_1+k_2-n}(\zeta+1)^{k_0+k_1-n}}{(\zeta-1)^{k_0+k_1-n}(\zeta+2)^{k_1+k_2-n}(2\zeta+1)^{k_1+k_3-n}}\right)^{n-1}\\
&\quad\times t^{(n-k_0,n-k_1,n-k_2,n-k_3)}(\zeta).
\label{tse}\end{align}
\end{corollary}

\begin{proof}
Proceeding as in the proof of Theorem \ref{trt} but starting from
the identity $\mathbf X\circ s_1s_3s_4=\mathbf X$ gives
 $\mathbf X(\tau_{l_1l_2l_3l_4})=\mathbf X(\tau_{-l_1,-l_0-l_2,-l_3,-l_4})$. Substituting \eqref{xtt} and writing
$$\frac{\phi_{l_1l_2l_3l_4}}{\phi_{-l_1,-l_0-l_2,-l_3,-l_4}}=(-1)^{l_1l_3}\left(\frac{(\zeta-1)^{l_3}(\zeta+2)^{l_4}(2\zeta+1)^{l_1}}{\zeta^{l_4}(\zeta+1)^{l_3}}\right)^{l_0-1} $$
we obtain \eqref{tse} after simplification.
\end{proof}

The symmetries
 \eqref{tsa}--\eqref{tsd} and \eqref{tse} 
generate the group
$G=\mathrm S_4\times \mathrm S_2\times\mathrm S_2$. This 
is the full set of symmetries arising from \eqref{xtsx}.
Indeed, the group generated by $s_0$, $s_1$, $s_3$, $s_4$, $t_1$ and $t_3$ under the relations \eqref{wr} is equal to $G$.

As another application, we can obtain further bilinear relations for  $t^{(\mathbf k)}$. Probably, any such relation can also be found using the method explained in \cite[\S 4]{r3} (see also \cite[\S 4.3]{zj}), that is, by combining minor relations for
the determinant defining  $T_n^{(\mathbf k)}$ with differential relations derived
from  \cite[Thm.\ 3.3]{r3}.  However, the approach based on B\"acklund transformations is more systematic. There are many such relations, 
but we will only give one example.

\begin{proposition}\label{nbp}
The functions $t^{(\mathbf k)}=t^{(k_0,k_1,k_2,k_3)}$ satisfy the bilinear relation
\begin{multline}\label{nbr}
-\frac{(2k_0+1)(2k_1+1)(\zeta+2)^2}{\zeta^2}\,t^{(\mathbf k+\mathbf e_0+\mathbf e_1)}t^{(\mathbf k-\mathbf e_0-\mathbf e_1)}\\
=A
\left(\frac{d^2t^{(\mathbf k)}}{d\zeta^2}t^{(\mathbf k)}-\left(\frac{dt^{(\mathbf k)}}{d\zeta}\right)^2\right)+B\frac{dt^{(\mathbf k)}}{d\zeta}t^{(\mathbf k)}+\frac{C}{4}\left(t^{(\mathbf k)}\right)^2,
\end{multline}
where $\mathbf e_j$ are unit vectors and
\begin{align*}
A&=\zeta(\zeta+1)^2(\zeta-1)^2(\zeta+2)(2\zeta+1),\\
B&=2(\zeta+1)^2(\zeta-1)(\zeta^3-3\zeta^2-6\zeta-1),
\end{align*}
\begin{multline*}
C=(39\zeta^4+110\zeta^3+116\zeta^2+50\zeta+9)k_0^2+(35\zeta^4+110\zeta^3+124\zeta^2+50\zeta+5)k_1^2\\
\shoveleft{ +(31\zeta^4+70\zeta^3+32\zeta^2-14\zeta-11)k_2^2+(19\zeta^4+46\zeta^3+32\zeta^2+10\zeta+1)k_3^2} \\
\shoveleft{ +2(29\zeta^4+110\zeta^3+136\zeta^2+50\zeta-1)k_0k_1+2(\zeta-1)(35\zeta^3+93\zeta^2+87\zeta+25)k_0k_2}\\
\shoveleft{+2(\zeta-1)^2(5\zeta^2+8\zeta+5)k_0k_3+2(\zeta-1)(9\zeta^3+19\zeta^2+17\zeta+3)k_1k_2}\\
\shoveleft{+2(\zeta-1)(27\zeta^3+73\zeta^2+71\zeta+21)k_1k_3
+2(17\zeta^4+58\zeta^3+48\zeta^2-2\zeta-13)k_2k_3}\\
\shoveleft{-2(3\zeta^4-52\zeta^3-136\zeta^2-112\zeta-27)k_0+2(\zeta^4+52\zeta^3+128\zeta^2+112\zeta+31)k_1}\\
\shoveleft{-2(27\zeta^4+68\zeta^3+44\zeta^2-16\zeta-15)k_2-2(15\zeta^4+44\zeta^3+44\zeta^2+8\zeta-3)k_3}\\
\shoveleft{+8(\zeta+1)^2(\zeta+2)(2\zeta+1).\hfill}\end{multline*}
\end{proposition}

\begin{proof}Let $\psi$ be the prefactor in \eqref{dz}, so that $\delta=\psi\cdot d/d\zeta$ on $\mathbb C(\zeta)$. Substituting \eqref{xtt} in \eqref{tbr}, 
using
$$\frac{\phi_{l_1,l_2+1,l_3-1,l_4}\phi_{l_1,l_2-1,l_3+1,l_4}}{\phi_{l_1l_2l_3l_4}^2}=\frac{(-1)^{l_3+l_4}\ti}{16u\zeta^2(2\zeta+1)^2}, $$
we find that \eqref{nbr} holds with
$$A=4(\zeta+2)^2(2\zeta+1)^2\frac{\psi^2}{t},\qquad B=4(\zeta+2)^2(2\zeta+1)^2\psi\left(\frac 1 t\frac{d\psi}{d\zeta}-1\right),$$
$$ C=16(\zeta+2)^2(2\zeta+1)^2\left(S(l_0,l_1,l_3,l_4)+\frac \psi t\frac{d}{d\zeta}
\left(\frac{\delta(\phi_{l_1l_2l_3l_4})}{\phi_{l_1l_2l_3l_4}}\right)-\frac{\delta(\phi_{l_1l_2l_3l_4})}{\phi_{l_1l_2l_3l_4}}\right). $$
Using \eqref{dpn}, one may check that this agrees with the given expressions.
\end{proof}

Proposition \ref{nbp} settles some conjectures for polynomials related to
solvable models. 
In \cite{bm1}, Bazhanov and Mangazeev found that the ground state eigenvalue for the $Q$-operator of a certain XYZ chain
can be expressed in terms of special polynomials $\mathcal P_n(x,z)$.
In \cite{bm2}, it was conjectured that, as a polynomial in $x$, the highest and
lowest coefficients of $\mathcal P_n$ are Painlev\'e tau functions.
In \cite[\S 5]{r2}, we showed that those coefficients are essentially $t^{(n,n,0,0)}$ and $t^{(n,n,1,-1)}$. 
Thanks to Theorem~\ref{trt}, this interesting relation between Painlev\'e VI and
the eight-vertex model is now rigorously established. 
We can then obtain the recursions of
\cite[Conj.~1(b)]{bm2} as special cases of
Proposition \ref{nbp}. For instance, substituting $\mathbf k=(n,n,0,0)$ in 
\eqref{nbr}, we find that $t_n=t^{(n,n,0,0)}$ satisfies
\begin{equation}\label{btr}
-\frac{(2n+3)(2n+1)(\zeta+2)^2}{\zeta^2}\,t_{n+1}t_{n-1}=A
\left(t_nt_n''-(t_n')^2\right)+Bt_n't_n+D_nt_n^2,
\end{equation}
with
\begin{align*}D_n&=(33\zeta^4+110\zeta^3+128\zeta^2+50\zeta+3)n^2-(\zeta^4-52\zeta^3-132\zeta^2-112\zeta-29)n\\
&\quad+2(\zeta+1)^2(2\zeta+1)(\zeta+2). \end{align*}

In \cite{r1}, we showed that
the partition function for the three-colour model with domain wall boundary conditions can be expressed in terms of certain polynomials $p_n$,
which are essentially equal to  $t^{(n+1,n,0,-1)}$ \cite[Eq.\ (5.5)]{r2}. 
We find from  \eqref{nbr} that $t_n=t^{(n+1,n,0,-1)}$ satisfies \eqref{btr} with
\begin{align*}D_n&=(33\zeta^4+110\zeta^3+128\zeta^2+50\zeta+3)n^2+(17\zeta^4+140\zeta^3+262\zeta^2+188\zeta+41)n\\
&\quad+2(11\zeta^4+53\zeta^3+79\zeta^2+47\zeta+8). \end{align*}
This proves \cite[Conj.\ 6]{mb}.

The functions $t^{(0,2n,0,0)}$ and $t^{(-1,2n+1,0,0)}$ seem to
appear in connection with eigenvectors of the Hamiltonian of the
 XYZ chain \cite{mb,ras,zj} and other spin chains
 \cite{bh,h}, though these relations have not yet been established rigorously.
Partial results were obtained by 
Zinn-Justin \cite{zj}, who also
derived recursions for these functions. 
One can give alternative proofs of those recursions using the
relation to Painlev\'e tau functions. In fact, one can derive a 
general  relation of the form
\begin{multline*}
t^{(\mathbf k+2\mathbf e_1)}t^{(\mathbf k-2\mathbf e_1)}
=A
\left((2k_1+1)^2\frac{d^2t^{(\mathbf k)}}{d\zeta^2}t^{(\mathbf k)}-(2k_1+3)(2k_1-1)\left(\frac{dt^{(\mathbf k)}}{d\zeta}\right)^2\right)\\
+B\frac{dt^{(\mathbf k)}}{d\zeta}t^{(\mathbf k)}+C\left(t^{(\mathbf k)}\right)^2.
\end{multline*}
The coefficients are  more complicated than
for \eqref{nbr}, and we do not go into the details.

% To give an example, 
% the polynomials $t_n=t^{(0,2n,0,0)}$ satisfy
% \begin{multline*}
% (\zeta+2)^2(4n-1)(4n+1)^2(4n+3)t_{n+1}t_{n-1}\\
% ={(4n+1)^2f^2}t_nt_n''
% -(4n-1)(4n+3)f^2(t_n')^2+2fE_nt_nt_n'+F_nt_n^2,
% \end{multline*}
% where
% \begin{align*}
% f&=\zeta(\zeta-1)(\zeta+1)(\zeta+2)(2\zeta+1),\\
% E_n&=-4(\zeta+2)(3\zeta^2+2\zeta+1)n^2+4(5\zeta^4+13\zeta^3-7\zeta^2-15\zeta-5)n\\
% &\quad-(5\zeta^2+5\zeta+2),\\
% F_n&=4(16\zeta^8+112\zeta^7+421\zeta^6+784\zeta^5+838\zeta^4+556\zeta^3+181\zeta^2+4\zeta+4)n^4\\
% &\quad-8(12\zeta^8+75\zeta^7-18\zeta^6-420\zeta^5-648\zeta^4-420\zeta^3-78\zeta^2+33\zeta+6)n^3\\
% &\quad+(40\zeta^8+196\zeta^7+379\zeta^6-26\zeta^5-8\zeta^4+730\zeta^3+709\zeta^2+220\zeta+28)n^2\\
% &\quad-(8\zeta^8+44\zeta^7+101\zeta^6+152\zeta^5+104\zeta^4-4\zeta^3-49\zeta^2-28\zeta-4)n\\
% &\quad-6\zeta^2(\zeta+1)^2(\zeta^2+\zeta+1).
% \end{align*}

% Finally, special cases of the polynomials $t^{(\mathbf k)}$ seem to appear
% in connection with eigenvectors of the Hamiltonian of the XYZ spin chain, again with $\Delta=-1/2$, and also for some related models. A rigorous study of the supersymmetric XYZ chain was recently been initiated by Zinn--Justin \cite{}.
% In particular

Using \eqref{pph} and \eqref{xtt} in Proposition \ref{qdp}, we find that 
$t^{(k_0,k_1,k_2,k_3)}$ always satisfies a 
quadratic differential equation. 
This seems to be a new observation.

\begin{proposition}\label{qdp2}
The polynomial $t=t^{(k_0,k_1,k_2,k_3)}(\zeta)$ satisfies a  differential equation of the form
\begin{equation}\label{qdt}\sum_{i\geq j\geq 0,\ i+j\leq 4}A_{ij}\frac{d^it}{d\zeta^i}\frac{d^jt}{d\zeta^j}=0 \end{equation}
with coefficients $A_{ij}$ that are polynomials in $\zeta$ and $k_0,\dots,k_3$.
\end{proposition}

One may normalize \eqref{qdt} so that
\begin{align*}
A_{40}&=e^3,\qquad A_{31}=-4e^3,\qquad A_{22}=3e^3,\\
A_{30}&=4\zeta^2(\zeta+1)^2(\zeta-1)^3(\zeta+2)^3(2\zeta+1)^4,
\end{align*}
where
$$e=\zeta(\zeta+1)(\zeta-1)(\zeta+2)(2\zeta+1).$$
The remaining coefficients depend on  $k_j$ and  are
 too cumbersome to write down; for instance, $A_{00}$ has 579 terms.

As an example, we consider the case of
$t=t^{(0,2n,0,0)}$. It follows from \eqref{ot} and \eqref{tsc}--\eqref{tsd} that
$$t^{(0,2n,0,0)}(\zeta)=\left(\frac{\zeta(\zeta+1)}{\zeta+2}\right)^{n(n-1)}f_n\left((2\zeta+1)^2\right), $$
with $f_n$  a polynomial of degree  $n(n-1)/2$. 
It is related to the polynomial
$q_n$ of \cite{mb} by
$$q_n(z)=D_n\,z^{n(n-1)}f_n(z^{-2}), $$
where $D_n$ is a constant, see \cite[\S 5.3]{r2}.
In terms of $f_n(z)$, \eqref{qdt} takes the form
\begin{multline*}
z(z-1)^3(z-9)^3\big(f_n^{(4)}f_n-4f_n^{(3)}f_n'+3(f_n'')^2\big)\\
\shoveleft{+(7z-3)(z-1)^2(z-9)^3\big(f_n^{(3)}f_n-f_n''f_n'\big)}\\
\shoveleft{-2(z-1)(z-9)\big\{(z+1)(z-9)^2n^2+2(z-9)^2n-5z^3+105z^2-483z+351 \big\}f_n''f_n}\\
\shoveleft{+2(z-1)(z-9)\big\{(z+1)(z-9)^2n^2+2(z-9)^2n-z^3+9z^2-111z+135\big\}(f_n')^2}\\
\shoveleft{-\big\{2(z-9)(z^3-39z^2+139z+27)n^2+8(z-9)(3z^2+2z+27)n}\\
\shoveleft{-2z^4+72z^3-876z^2+2184z-1890\big\}f_n'f_n}\\
\shoveleft{-2n(n-1)\big\{(5z-21)(z-9)n^2-(z+15)(z-9)n+z^2+22z+9\big\}f_n^2=0}.
\end{multline*}

As a final remark,
we stress that the functions $t_n^{(k_0,k_1,k_2,k_3)}$ are defined by
 explicit determinants. For instance, writing
 $a=2\zeta+1$, $b=\zeta/(\zeta+2)$
and 
$$G(x,y)=(\zeta+2)xy(x+y)-\zeta(x^2+y^2)-2(\zeta^2+3\zeta+1)xy+\zeta(2\zeta+1)
(x+y),$$
we have
\begin{align}\notag t^{(n,n,0,0)}&=\lim_{\substack{x_1,\dots,x_n\rightarrow a\\y_1,\dots,y_n\rightarrow b}}
\frac{\prod_{i,j=1}^nG(x_i,y_j)}{\prod_{1\leq i<j\leq n}(y_j-y_i)(x_j-x_i)}\,\det_{1\leq i,j\leq n}\left(\frac{1}{G(x_i,y_{j})}\right),\\
\label{ntd}&=\frac{G(a,b)^{n^2}}{\prod_{j=1}^{n}(j-1)!^2}
\,\det_{1\leq i,j\leq n}\left(\frac{\partial^{i+j-2}}{\partial x^{i-1}\partial y^{j-1}}\Bigg|_{x=a,y=b}\frac{1}{G(x,y)}\right).
\end{align}
These functions solve the recursion \eqref{btr}.
This is reminiscent of how the Toda equation
\begin{equation}\label{te}\tau_{n+1}\tau_{n-1}=\tau_n''\tau_n-(\tau_n')^2 
\end{equation}
is solved by Hankel determinants
\begin{equation}\label{tes}\tau_n=\det_{1\leq i,j\leq n}(f^{(i+j-2)}). 
\end{equation}
However, an important difference is that, whereas \eqref{te} is immediately obtained
from \eqref{tes} by applying the Jacobi--Desnanot identity,
applying that identity to \eqref{ntd}
leads to an equation involving  $x$- and $y$-derivatives of $G(x,y)$, cf.\ \cite[Cor.~7.16]{r1}.
The missing ingredient is  the Schr\"odinger equation (or \emph{quantum} Painlev\'e VI equation) of \cite{r3}, which allows us to express
specialized $x$- and $y$-derivatives of $G$ in terms of $\zeta$-derivatives.

It should be mentioned that genuine Hankel determinants for tau functions of Painlev\'e VI have been given in \cite{ny2}. These are  quite
  different
in nature from \eqref{ntd}. It would be interesting to know whether
identities such as \eqref{ntd} are peculiar to our choice of seed solution, or if similar formulas can be found for other solutions.

 \end{document}